\documentclass[12pt]{amsart}
\usepackage{amssymb, euscript, mathrsfs, mathtools}
\usepackage{upref}
\usepackage[caption=false]{subfig}
\usepackage{graphicx}
\usepackage[margin = 0.95in]{geometry}

\newcommand{\vk}{\varkappa}

\newcommand{\BR}{\mathbb{R}}

\newcommand{\SL}{\sum\limits}

\newcommand{\al}{\alpha}

\newcommand{\de}{\delta}

\newcommand{\ME}{\mathbb E}

\newcommand{\MP}{\mathbf P}
\newcommand{\CE}{\mathcal E}

\newcommand{\CL}{\mathcal L}

\newcommand{\CK}{\mathcal K}

\newcommand{\si}{\sigma}

\newcommand{\pa}{\partial}
\renewcommand{\phi}{\varphi}

\newcommand{\la}{\lambda}

\newcommand{\ol}{\overline}

\newcommand{\CM}{\mathcal M}

\newcommand{\norm}[1]{\lVert#1\rVert}

\newcommand{\md}{\mathrm{d}}

\DeclareMathOperator{\diag}{diag}
\DeclareMathOperator{\const}{const}
\DeclareMathOperator{\tr}{tr}

\DeclareMathOperator{\mes}{mes}

\begin{document}

\theoremstyle{plain}
\newtheorem{theorem}{Theorem}[section]
\newtheorem{lemma}[theorem]{Lemma}
\newtheorem{prop}[theorem]{Proposition}
\newtheorem{cor}[theorem]{Corollary}

\theoremstyle{definition}
\newtheorem{defn}{Definition}
\newtheorem{asmp}{Assumption}
\newtheorem{remark}{Remark}
\newtheorem{exm}{Example}

\title[Interbank flows, Borrowing, and Investing]{Modeling Financial System with Interbank Flows,\\ Borrowing, and Investing}

\author{Aditya Maheshwari}

\address{\scriptsize Department of Statistics and Applied Probability, University of California, Santa Barbara}

\email{aditya\_maheshwari@ucsb.edu}

\author{Andrey Sarantsev}

\address{\scriptsize Department of Mathematics and Statistics, University of Nevada, Reno} 

\email{asarantsev@unr.edu}

\begin{abstract}
In our model, private actors with interbank cash flows similar to, but nore general than (Carmona, Fouque, Sun, 2013) borrow from the outside economy at a certain interest rate, controlled by the central bank, and invest in risky assets. Each private actor aims to maximize its expected terminal logarithmic utility. The central bank, in turn, aims to control the overall economy by means of an exponential utility function. We solve all stochastic optimal control problems explicitly. We are able to recreate occasions such as liquidity trap. We study distribution of the number of  defaults (net worth of a private actor going below a certain threshold).
\end{abstract}

\date{\today. Version 26}

\subjclass[2010]{60H10, 60H30, 60J60, 60K35, 91A15, 91B70, 91G80, 93E15}

\thanks{\textit{2017 JEL Classification.} C61, E43, E44, E52, G11. \\ \indent The first and the second author are supported in part by NSF grants DMS 1736439 and DMS 1409434.\\ \indent We are grateful to \textsc{Rene Carmona}, \textsc{Zachary Feinstein}, \textsc{Jean-Pierre Fouque}, \textsc{Matheus Grasselli}, \textsc{Ioannis Karatzas}, \textsc{Andreea Minca}, and \textsc{Soumik Pal} for useful discussion.}

\keywords{Systemic risk, stochastic control, principal-agent problem, stochastic game, stationary distribution, stochastic stability, Lyapunov function}

\maketitle

\thispagestyle{empty}

\section{Introduction}

We are interested in modeling interaction between utility-maximizing private actors (which for simplicity we call {\it private banks} or simply {\it banks}) and a central bank, which regulates borrowing activity via an interest rate. Private banks exchange (exogenous) cash flows, and borrow from the general economy to invest in profitable but risky assets. This central bank can lower interest rate to stimulate financial activity by private actors, or increase this rate to cool this activity down. Sometimes, however, there are not many profitable investments. Then the private actors do not borrow at all, while the central bank is not able to remedy this even by lowering the rate to zero; this is called {\it liquidity trap}. 

\smallskip

We use the utility maximization approach: private actors borrow and invest to maximize their logarithmic utility, and the central bank applies its own exponential (which is sometimes referred to as CARA = constant absolute risk aversion) utility function to a certain variable which measures the size of the entire financial system. The utility functions are chosen so that the central bank (interested in stability of the system overall) is more risk-averse than private actors (interested only in their own net worth). Because of the special choice of utility functions, we are able to solve all corresponding Hamilton-Jacobi-Bellman equations explicitly. 

\smallskip

Logarithmic utility, as shown in \cite{Myopic}, corresponds to myopic decision-making; in other words, private actors are short-sighted. The central bank is more risk-averse, and sometimes it needs to reduce overall risk by reducing the interest rate. 

\smallskip

We mention the concept of {\it systemic risk}, which can be informally described as the probability of a large number of banks defaulting or getting into financial trouble. We understand {\it default} or {\it failure} of a bank as its net worth (assets minus liabilities) going below a given threshold. We are interested in probability of this undesirable event; of the mechanism of such failure; and of the {\it financial contagion}, when failure of a few banks leads to many more failures. We refer the reader to the handbook \cite{Handbook} containing many different approaches to systemic risk. Our work is inspired by the model introduced in \cite{JP-Rene} and also described in \cite[Section 5.5]{ReneBook}. 

\smallskip

If $X_i(t)$ is the wealth and $Y_i(t) := \log X_i(t)$, authors in \cite{JP-Rene} model the banking system as a system of $N$  continuous-time stochastic processes $Y_1, \ldots, Y_N$, with multidimensional Ornstein-Uhlenbeck dynamics. The stochastic differential equations are given by:
\begin{equation}
\label{eq:basic-mean-field}
\md Y_i(t) = a\left(\ol{Y}(t) - Y_i(t)\right)\,\md t + \sigma\md W_i(t),\ \ i = 1, \ldots, N,
\end{equation}
with i.i.d. Brownian motions $W_1, \ldots, W_N$, constants $a, \sigma > 0$, and
\begin{equation}
\label{eq:mean}
\ol{Y}(t) = \frac1N\SL_{i=1}^NY_i(t),\quad t \ge 0.
\end{equation}

The constant $a$ is referred to as the \textit{interbank flow rate}. In \cite{JP-Rene}, these mean-reverting drifts are generated by the decisions of banks to borrow money from one another. Their decisions are done by minimizing a certain cost functional, which measures, roughly speaking, the preference of a bank to borrow from other banks, as opposed to borrowing from the central bank.  Authors  discuss both the finite player solution and the mean-field limit of the problem in the context of systemic risk. 

\smallskip

An important observation: Apply It\^o's formula to rewrite equations~\eqref{eq:basic-mean-field} in terms of $X_i$ (the actual net worth of the $i$th bank) instead of $Y_i = \log X_i$. Then the interbank flows derived from Ornstein-Uhlenbeck-type terms $a(\ol{Y}(t) - Y_i(t))\,\md t$ do not add up to $0$. One can think that the remainder comes from (or to, depending on the sign) the real economy. Nevertheless, the model~\eqref{eq:basic-mean-field} attracted a lot of attention because of its simplicity and analytical tractability. In this article we shall build on it, extending these Ornstein-Uhlenbeck-type terms to be heterogeneous; see below.

\smallskip

We further explore the individual decision-making of the private banks and extend the role of the central bank. Furthermore, we analyze how this decision-making affects the stability of the system. More specifically, we extend the model by assuming that each private bank invests in a risky portfolio of assets, borrows money from the general economy to invest in this portfolio (with interest rate controlled by the central bank), pockets the profit, and pays back the interest. 

\smallskip

Unlike \cite{JP-Rene}, where the decision making of the central bank is not analyzed, in our model the central bank uses interest rate as a policy tool (to govern the behavior of the private banks). It is derived as a solution of the control problem solved by the central bank.

\smallskip

We assume two kinds of players in our model: private banks and the central bank. Private banks want to maximize the terminal logarithmic wealth:
\begin{equation}
\label{eq:optimal-private-bank}
\sup\ME\left[\log X_i(T)\right] = \sup\ME\left[Y_i(T)\right],\, i = 1, \ldots, N, 
\end{equation}
through borrowing and investing in the portfolio of risk assets. For simplicity, we assume the portfolios of private bank to be correlated  geometric Brownian motions. The choice of logarithmic utility function allows us to solve the corresponding Hamilton-Jacobi-Bellman (HJB) equation explicitly. 

\smallskip

On the other hand, the central bank wants to control the overall size of the financial system, using {\it interest rate} $r > 0$ as a monetary policy instrument, which measures how attractive it is for banks to borrow and invest, as opposed to sitting on cash. As we see in Section 4, this interest rate $r$ controls the overall size of the system, measured by $\ol{Y}$ from~\eqref{eq:mean}.
This is {\it not} the average net worth of banks; this is the average of the logarithms of net worth. This measure is somewhat non-standard; however, it is more appropriate for our model, when dynamics in~\eqref{eq:basic-mean-field} is written in terms of logarithms $Y_i$ of net worths $X_i$. This measure is used in \cite{JP-Rene} and subsequent papers, so we feel justified in using it here. The central bank chooses $r$ to maximize terminal expected exponential utility (sometimes in the literature it is called CARA: constant relative risk aversion):
\begin{equation}
\label{eq:optimal-central-bank}
\ME \left[-\exp(-\la \ol{Y}(T))\right]\ \mbox{for some}\ \la > 0.
\end{equation}
This corresponds to the central bank being even more risk-averse than private banks: Private banks have utility function which is linear in $Y_i(t)$, and the central bank has utility function which is concave in these variables. As for the private actors, we can solve the corresponding Hamilton-Jacobi-Bellman equation explicitly and find the optimal interest rate $r$. This is due to the special choice of 

\smallskip

Given the choice of the interest rate $r$, we can solve the stochastic optimal control problem explicitly for each player: private banks and the central bank. This is due to the special choice of logarithmic utility function for private banks in~\eqref{eq:optimal-private-bank}, and for the central bank in~\eqref{eq:optimal-central-bank}. For other choices of utility function, it is probably impossible to solve this optimal control problem explicitly. Then one could try to use mean-field limits, as in \cite{Lacker3, Lacker2, Lacker1}. This topic is left for future research. 

\smallskip

This setup somewhat resembles the principal-agent problem: the principal (now the central bank) allows private banks to borrow from  the economy, and private banks (agents) maximize their expected logarithmic terminal utility (their contract). 

\smallskip

Under such optimal choices of the actors, we study the dynamics of logarithmic net worth of banks, and the distribution of defaults. A default of the $i$th bank is understood in the same way as above: when $X_i(t)$, the net worth of this bank, falls below some fixed positive  threshold. This leads us to understanding systemic risk in this model: how defaults of a few banks can lead to defaults of many other banks, see subsection 3.4. 

\smallskip

Besides incorporating the optimal strategy of the central bank, we also generalize the model~\eqref{eq:basic-mean-field} by allowing interbank flow rates from bank $i$ to bank $j$ to depend on the banks $i, j$, and on time $t$, denoting this rate by $c_{ij}(t)$:
\begin{equation}
\label{eq:enhanced-mean-field}
\md Y_i(t) = \frac1N\SL_{j=1}^Nc_{ij}(t)\left(Y_j(t) - Y_i(t)\right)\,\md t + \md W_i(t).
\end{equation}
Here, we assume that the flow rates satisfy 
$$
c_{ij}(t) = c_{ji}(t),\ i \ne j;\ c_{ii}(t) = 0,\ i = 1, \ldots, N.
$$
This heterogeneity, together with Ornstein-Uhlenbeck dynamics, resembles to some extent the model by \cite{Kluppelberg2}. As we see in Section 4, the matrix $(c_{ij}(t))$ of interbank flow rates corresponds to the stability of the system. We also allow for Brownian motions $W_1, \ldots, W_N$ to have drifts and to be correlated: that is, we assume $W = (W_1, \ldots, W_N)$ is an $N$-dimensional Brownian motion with drift vector $\mu$ and covariance matrix $A$. 

\smallskip

The observation made above for~\eqref{eq:basic-mean-field} can be applied for~\eqref{eq:enhanced-mean-field}: If we rewrite these equations using It\^o's formula in terms of $X_i$ (the actual net worth) instead of $Y_i = \log X_i$, then the interbank flows do not add up to $0$. As before, one can think that the remainder comes from the real economy. But we shall build on the model~\eqref{eq:basic-mean-field} from \cite{JP-Rene}, which attracted a lot of research interest because of its tractability and lucidity.

\smallskip

To summarize, the topic of this paper is optimal decisions of private banks, and interaction of these individual decisions with each other, as well as with that of the central bank. Systemic risk is a consequence of the optimal decision making of both the central bank and the private banks in the economy.


\subsection{Contributions}  Our work here is inspired by \cite{JP-Rene}, however, unlike them, we consider utility maximization for both the central bank and the private banks.  We first solve for the optimal investment of the private banks in their portfolios of assets, given the interest rate set by the central bank. Next, we turn around and find the optimal interest rate for the central bank, given the strategy of the private banks and its own risk aversion. Unlike in \cite{JP-Rene}, where the interbank flow rate was constant, here we generalize the flow rates to be different for each pair of private banks. We consider dependence of the distribution of defaults on the correlation between risky assets, and on the interest rate set by the central bank. 

\subsection{Organization of the paper} In Section 2, we describe the model in terms of stochastic control problem for a system of stochastic differential equations. In Section 3, we solve the stochastic control problem for each private bank, and in Section 4, for the central bank (given optimal control for each private bank). In particular, in Section 3, we study distribution of the number of defaults. This is where we touch the concept of systemic risk: We are interested in its dependence on the parameters of the system, for example correlations between various risky investments. Section 5 contains results on long-term stability of the system: the fact that the capitals of banks tend to stay close, as opposed to splitting into two or more groups. Section 6 is devoted to concluding remarks and suggestions for future research. The Appendix contains some technical proofs. 

\subsection{Review of related models} 
The fundamental wealth dynamics~\eqref{eq:basic-mean-field} of this model has been studied under different settings. For example, in \cite{Delay}, a similar system was studied with time delay. Large deviations were studied in \cite{Garnier}. Without attempting to give an exhaustive survey, let us mention the following papers, which use stochastic differential equations and interacting Brownian particles to model dynamics of capitals of banks or other financial agents. \cite{Award}  use a system of stochastic differential equations with Bessel-type diffusion coefficients  to model simultaneous defaults (in this model, a {\it default} is when the capital reaches zero). Authors in \cite{Bo, Sun} combine such Bessel-type diffusion coefficients with a mean-field-type drift term, with \cite{Bo} having an additional jump term (therefore the processes there are jump-diffusions). 

\smallskip

In the paper \cite{Sergey}, the financial system is modeled by independent geometric Brownian motions, with defaults happening at hitting times of some lower threshold. Once a bank defaults, other banks see their capital decrease by a certain amount, possibly triggering a cascade of defaults. \cite{MFG} introduce {\it mean-field game of timing}. The term {\it game of timing} refers to a game where each player chooses an optimal stopping time. {\it Mean-field game of timing} is when a player competes against a ``crowd'' of other agents, instead of individual competitors; this can be informally viewed as a limit of games of timing as the number of players tends to infinity. This models a bank run, continuing the research in a celebrated paper \cite{Diamond}.

\smallskip

Let us also mention the paper  \cite{Thaleia}, which models the dynamics by geometric Brownian motions, studying mean-field relative performance criteria: An agent competes against a ``crowd''. We maximize the agent's performance compared to the performance of the ``crowd''. 

\smallskip

In a recent paper \cite{New1}, banks are organized in clusters.  The interbank transactional dynamics is modeled through a set of interacting measure-valued processes. Implications of shocks arising in a cluster are studied. 

\subsection{Notation} For a vector or a matrix $a$, its transpose is denoted by $a'$. We usually think of vectors as column-vectors. The dot product of two vectors $a$ and $b$ is denoted by $a\cdot b$. The term {\it standard Brownian motion} stands for a one-dimensional Brownian motion with drift coefficient $0$ and diffusion coefficient $1$. For $V \equiv 1$, this is called the {\it total variation norm}. Fix a dimension $N \ge 2$. Then $e \in \BR^N$ is a vector $(1, \ldots, 1)'$ with unit components, and we define the following hyperplane in $\BR^N$:
$$
\Pi := \{x \in \BR^N\mid x\cdot e = 0\} = \{x \in \BR^N\mid x_1 + \ldots + x_N = 0\}. 
$$
Define the (closed) ball of radius $r$ on $\Pi$ centered at the origin:
\begin{equation}
\label{eq:ball}
\mathcal B(r) := \{x \in \Pi\mid \norm{x} \le r\}. 
\end{equation}
The $(N-1)$-dimensional Lebesgue measure on $\Pi$ is denoted by $\mes_{\Pi}(\cdot)$. As mentioned above, the symbol $1(A)$ or $1_A$ stands for the indicator function of an event $A$. 


\section{Description of the model}

\subsection{Formal description} Consider a system of $N$ agents (we call them {\it private banks}) which continuously lend money to each other, borrow from the outside economy, pay back the interest, and invest in some risky portfolios.  

\smallskip

We operate on filtered probability space $(\Omega, \mathcal{F}, ( \mathcal{F}_t)_{t\ge0}, \mathbb{P})$ with the filtration satisfying the usual conditions. All the processes which we consider are adapted to the $( \mathcal{F}_t)_{t\ge0}$. Let $X_i(t) > 0$ be the net worth (assets minus liabilities) of the $i$th bank at time $t$, for $i = 1, \ldots, N$. Let $Z_i(t)$  be the amount borrowed at the moment $t$ by the $i$th private bank from the outside economy. Assume the interest rate for such borrowing is $r(t) \ge 0$, controlled by the central bank. Then during the time interval $[t, t + \md t]$, the $i$th bank pays back interest $r(t)Z_i(t)\,\md t$. At time $t$, the $i$th bank has at its disposal the amount $X_i(t) + Z_i(t)$: its own capital plus borrowed amount. This amount $Z_i(t) \ge 0$ is controlled by the $i$th bank.

\smallskip

Alternatively, the $i$th bank might decide to not borrow anything, and instead to even put aside some of its own money in cash (which does not earn any interest). This happens if the investment is not very profitable, or, more precisely, if the return does not outweigh the risk. In this case, we let $Z_i(t) < 0$, and define $-Z_i(t)$ to be the quantity of cash put aside. The amount invested is still $X_i(t) + Z_i(t)$, but the bank does not pay or receive any interest. 

\smallskip

We combine these two cases: the $i$th bank invests the amount $X_i(t) + Z_i(t)$ at time $t$ into a risky portfolio, and pays interest $r(t)(Z_i(t))_+\,\md t$ during the time interval $[t, t + \md t]$. 

\smallskip

At time $t$, the $i^{th}$ bank invests in a portfolio of risky assets with value $S_i(t)$. The $i$th bank buys $(X_i(t) + Z_i(t))/S_i(t)$ units of this portfolio. Net profit for the time interval $[t, t + \md t]$ is 
$$
(X_i(t) + Z_i(t))\frac{\md S_i(t)}{S_i(t)}.
$$

Combining all of the above, we get the following system of equations:
\begin{equation}
\label{eq:main-noOU}
\md X_i(t) = (X_i(t) + Z_i(t))\frac{\md S_i(t)}{S_i(t)} - r(t)(Z_i(t))_+\,\md t \ \ i = 1, \ldots, N, \text{ and } X_i(0)>0.
\end{equation}
Next, we make some assumptions on $S_i$, the dynamics of the portfolio processes. A separate question is how banks construct these portfolios out of stocks and other risky assets. This question is separate from the topic of this paper, and we shall not study it here. Instead, we assume that these are geometric Brownian motions. This assumption is very simplifying, but we believe it captures to some extent the features of portfolios. The processes
$$
M_i(t) = \int_0^t\frac{\md S_i(s)}{S_i(s)},\ i = 1, \ldots, N,
$$
form an $N$-dimensional Brownian motion with drift vector $\mu = (\mu_1, \ldots, \mu_N)$ and covariance matrix $A = (a_{ij})_{i, j = 1, \ldots, N}$.  In particular, each $M_i,\, i = 1, \ldots, N$, is a Brownian motion with drift coefficient $\mu_i$ and diffusion coefficient $\si_i^2 := a_{ii}$, so it can be represented as 
\begin{equation}
M_i(t) = \mu_i t + \si_iW_i(t),
\label{eqn:assetPrice}
\end{equation}
where $W_i$ is a one-dimensional standard Brownian motion. Although the portfolio process \eqref{eqn:assetPrice} is driven by only one Brownian motion, a more general representation: 
\begin{equation}
    \frac{\md S_i(t)}{S_i(t)} = \mu_i dt + \sum_{j=1}^{m}\sigma_{i,j}\md B_j(t),
\end{equation}
where $(B_1,\ldots,B_m)$ are Brownian motions, can also be considered in our framework. Since $\sigma_i\md W_i(t) := \sum_{j=1}^{m}\sigma_{i,j}\md B_j(t)$, but $W_i$ is also a Brownian motion, we fall back to the original portfolio process \eqref{eqn:assetPrice}. 

\smallskip

The covariance between Brownian motions $(W_1, \ldots, W_N)$ can be modeled in various ways. 

\smallskip

{\bf (1.a)} All $W_1, \ldots, W_N$, are independent. Then the matrix $A$ is diagonal: 
\begin{equation}
\label{eq:diagonal-A}
A = \diag(\si_1^2, \ldots, \si_N^2).
\end{equation}
This means that the portfolios of banks are independent. 

\smallskip

{\bf (1.b)} All $W_1, \ldots, W_N$, are the same: $W_1 = W_2 = \ldots = W_N$. This means that all banks, in fact, use the same portfolio, and they are perfectly correlated. Then it makes sense to let $\mu_1 = \ldots = \mu_N$ and $\si_1 = \ldots = \si_N$. 

\smallskip

{\bf (1.c)} An intermediate case: for some i.i.d. Brownian motions $\tilde{W}_i,\, i = 0, \ldots, N$, and some coefficients $\rho_0, \tilde{\rho}_0$ with $\rho^2_0 + \tilde{\rho}^2_0 = 1$ we have:
\begin{equation}
\label{eq:correlated-assets}
W_i(t) := \rho_0\tilde{W}_i(t) + \tilde{\rho}_0\tilde{W}_0(t),\ i = 1, \ldots, N. 
\end{equation}

One can also split $N$ banks into subsets and construct dependence as in Case (1.c) for each subset; portfolio processes corresponding to different subsets are assumed to be independent. 

\subsection{Main system of driving stochastic equations} Apply It\^o's formula to find the dynamics of $Y_i(t) := \log X_i(t)$:
\begin{equation}
\label{eq:log-0}
\md Y_i(t) = \frac{\md X_i(t)}{X_i(t)} - \frac{\md\langle X_i\rangle_t}{2X_i^2(t)}
\end{equation}
Combining~\eqref{eq:main-noOU} with~\eqref{eq:log-0}, we get our main stochastic equation, driving banks' wealth. For now, it does not contain interbank flows, which are Ornstein-Uhlenbeck-type drifts as in~\eqref{eq:basic-mean-field} or~\eqref{eq:enhanced-mean-field}:
\begin{equation}
\label{eq:log}
\md Y_i(t) = (1 + \al_i(t))\si_i\,\md W_i(t)  + h_i(\al_i(t), r(t))\,\md t.
\end{equation} 
Here we define the {\it relative investment ratio:}
$$
\al_i(t) = \frac{Z_i(t)}{X_i(t)},\ \ t \ge 0,\ \ i = 1, \ldots, N,
$$
and the following quantity:
\begin{equation}
\label{eq:new}
h_i(\al, r) := (1 + \al)\mu_i - \frac{\si_i^2}2(1 + \al)^2 - r\al_+\ \ \mbox{for}\ \ \al, r \in \BR.
\end{equation}
Finally, the $i$th bank also interacts with other banks, having cash flow in and out. In \cite{JP-Rene} and subsequent papers, this interaction is modeled by  Ornstein-Uhlenbeck-type drifts 
\begin{equation}
\label{eq:OU}
a(\ol{Y}(t) - Y_i(t))
\end{equation}
from~\eqref{eq:basic-mean-field}, with  $Y_i(\cdot) = \log X_i(\cdot)$. Here, we take  drifts 
\begin{equation}
\label{eq:OU-general}
\frac1N\sum\limits_{j=1}^Nc_{ij}(t)(Y_j(t) - Y_i(t))
\end{equation}
from~\eqref{eq:enhanced-mean-field}, which are more general than~\eqref{eq:OU}, and add them  to~\eqref{eq:log}. Note that in our model, as in \cite{JP-Rene}, the cash flows (in the original scale, not logarithmic one) do not necessarily add up to zero. Consider possible particular cases:

\smallskip

{\bf (2.a)} All $c_{ij}(t) \equiv 0$. Then there are no cash flows between banks.

\smallskip

{\bf (2.b)} All $c_{ij}(t) \equiv c(t) > 0$. For a constant $c$, this is the model from \cite{JP-Rene}. 

\smallskip

{\bf (2.c)} Let $G$ be a graph on vertices $\{1, \ldots, N\}$. Then 
\begin{equation}
\label{eq:graph-model}
c_{ij}(t) = c(t)1(i \leftrightarrow j)\ \mbox{for some}\ c(t) > 0.
\end{equation}

After superimposing these Ornstein-Uhlenbeck-type drifts from~\eqref{eq:OU-general} on top of~\eqref{eq:log}, our main driving equation takes the form
\begin{align}
\label{eq:log-JP}
\begin{split}
\md Y_i(t)  = (1 + \al_i(t))&\si_i\,\md W_i(t) + h_i(\al_i(t), r(t))\,\md t \\ & + \frac1N\SL_{j=1}^Nc_{ij}(t)\left(Y_j(t) - Y_i(t)\right)\,\md t,\ \ i = 1, \ldots, N.
\end{split}
\end{align}
Equation~\eqref{eq:log-JP} resembles the model from \cite{JP-Rene}. However, it also has significant differences: the volatility in~\eqref{eq:log-JP} can be controlled, unlike in \cite{JP-Rene}; and the drift coefficient in~\eqref{eq:log-JP} is a bit more complicated. For  {\it homogeneous rates:} $c_{ij}(t) \equiv c(t)$, the equation~\eqref{eq:log-JP} takes the form
\begin{align}
\label{eq:heterogeneous-main}
\begin{split}
\md Y_i(t) = (1 + \al_i(t))\si_i\,\md W_i(t) + h_i(\al_i(t), r(t))\,\md t  + c(t)(\ol{Y}(t) - Y_i(t))\,\md t,
\end{split}
\end{align}
for $i = 1, \ldots, N$, where $\ol{Y}(t)$ is defined in~\eqref{eq:mean}. 

\subsection{Interpretation} As in \cite{JP-Rene}, we consider bank $i$ to be in bankruptcy at time $t$ if $X_i(t) < e^D$, where $D$ is a given threshold, stipulated by the central bank. The central bank would like to stimulate the activity of banks by persuading them to take risks, but not too much, lest they may become bankrupt. Equation~\eqref{eq:log} means that the central bank can use interest rate $r(t)$ as a monetary policy tool to alter the behaviour of the private-banks.

\smallskip

Assume that banks start borrowing too much money and investing them in risky assets (leveraging). By doing this, they increase their probability of default. Then the central bank can raise this interest rate to discourage private banks from excessive borrowing. Conversely, if banks are too cautious in borrowing against future profits and risk-taking, then the central bank can stimulate them by lowering the interest rate. As we see later, this interest rate affects the overall state of the system.

\smallskip

The parameter $r(t)$ is determined by the central bank and given to all private banks. These private banks then determine the investment rates $\al_i(t)$, independently of each other. In light of decision-making of the banks, the central bank needs to determine optimal values of these parameters. This setup is similar to the principal-agent problem, but with many agents. 

\section{Optimal behavior of private banks} 

\subsection{Statement of the problem} We assume that the $i$th bank takes as given the  capital of other banks: $X_j(t),\, j \ne i$ (or, equivalently, $Y_j(t) := \log X_j(t),\, j \ne i$), as well as  the interest rate $r(t)$ (the instrument of the monetary policy). The bank is trying to choose the relative investment ratio $\al_i(t)$, or, equivalently, the amount borrowed $Z_i(t)$, to maximize its expected terminal logarithmic wealth:
\begin{equation}
 \sup\limits_{\al_i}\, \ME\left[\log X_i(T)\right],
\label{eqn:opt_problem}    
\end{equation}
where the supremum in~\eqref{eqn:opt_problem} is taken over all bounded adapted controls $\al_i = (\al_i(t),\, 0 \le t \le T)$.  Assume that the interest rate $r(t)$ is already set by the central bank. In this section, we solve this stochastic control problem explicitly. This corresponds to the agent's problem in the principal-agent framework. In the next section, we discuss the optimal policy choices of the central bank (the principal). 

\subsection{Solution of the problem} This specific choice of the utility function, which is linear in $Y_i$, and the interbank flows, which are also linear in $Y_i$ in~\eqref{eq:log}, makes this optimization problem tractable: We can solve this explicitly. 

\begin{theorem} For the optimization problem stated in \eqref{eqn:opt_problem}, where $\al_i$ is bounded, adapted on $[0,T]$, the following value of $\al_i$ is optimal for the $i$th private bank:
\begin{equation}
\label{eq:alpha-max}
\al_i^*(t) := 
\begin{cases}
\left(\frac{\mu_i - r(t)}{\si_i^2} - 1\right)_+,\ \ \mu_i \ge \si_i^2;\\
\frac{\mu_i}{\si^2_i} - 1,\ \ \mu_i \le \si_i^2.
\end{cases}
\end{equation}
\label{thm:agent}
\end{theorem}

\begin{remark}  In particular, if $\mu_i \le \si_i^2$, that is, the return on the investment does not outweigh its risks, then the $i$th bank does not borrow anything to invest. On the contrary, this bank sets aside money as cash. If $\mu_i \ge \si_i^2$, the investment is attractive for borrowing, but a high enough interest rate: $r(t) \ge \mu_i - \si_i^2$ can preclude the $i$th bank from borrowing; then this bank will invest only its own money into the portfolio. Only if the interest rate is low enough: $r(t) < \mu_i - \si_i^2$, the $i$th bank borrows money to invest. 
\end{remark}

\begin{remark}
Note that the optimal strategies~\eqref{eq:alpha-max} do not depend on the flow rates $c_{ij}$, because of the special choice of logarithmic utility function, which is linear in $Y_i$. Although logarithmic utility function leads to myopic agents, this assumption is important for mathematical tractability of the results. CRRA utility function is another popular form used in the literature, however we were unable to evaluate the optimal control for it even in the mean field case. 
\label{remark:independent-of-flows}
\end{remark}

\begin{proof} The dynamic programming principle tells us that the function 
$$
\Phi_i(t, y) := \sup_{\al_i}\ME\left[Y_i(T)\mid Y_i(t) = y\right]
$$
where we take the supremum over all $\al_i$ which are bounded and adapted on $[t,T]$, satisfies the Hamilton-Jacobi-Bellman (HJB) equation:
\begin{align}
\label{eq:HJB-private}
\begin{split}
\frac{\pa\Phi_i}{\pa t}&(t, y) +  \sup\limits_{\al_i \in \BR}\biggl[\frac12\SL_{j=1}^N\SL_{k=1}^N(1 + \al_j)a_{jk}\frac{\pa^2\Phi_i}{\pa y_j\pa y_k}(t, y) \\ & + \SL_{j=1}^N\Bigl[h_j(\al_j, r(t)) + \frac{1}{N}\SL_{k=1}^Nc_{jk}(t)(y_k - y_j)\Bigr]\frac{\pa\Phi_i}{\pa y_j}(t, y)\biggr]
 = 0,
\end{split}
\end{align}
with terminal condition $\Phi_i(T, y) = y_i$. We assume all $\al_j,\, j \ne i$, are already chosen. Try the following Anzats, linear in $y_j$:
\begin{equation}
\label{eq:anzats}
\Phi_i(t, y) = g_{i0}(t) + \SL_{j=1}^Ng_{ij}(t)y_j.
\end{equation}
Because it is linear, the second-order derivatives in~\eqref{eq:HJB-private} turn out to be zero. Therefore, the only term in~\eqref{eq:HJB-private} which needs to be maximized is $h(\al_i, r(t))$. The solution to this maximization problem is given by the value $\al_i^*$ from~\eqref{eq:alpha-max}. This is a simple algebraic exercise; detailed calculations are given in Lemma \ref{lemma:solutionAlpha} in the Appendix. 


The maximal value of $h_i(\al, r(t))$ is 
\begin{equation}
\label{eq:h-max}
h_i^*(t) := h_i(\al_i^*(t), r(t)) = 
\begin{cases}
r(t) + \frac{(\mu_i - r(t))^2}{2\si_i^2},\ \ r(t) \le \mu_i - \si_i^2;\\
\mu_i - \frac12\si_i^2,\ \ r(t) \ge \mu_i - \si_i^2 \ge 0;\\
\frac{\mu_i^2}{2\si_i^2},\ \ \mu_i \le \si_i^2.
\end{cases}
\end{equation}
This means that the $i$th bank chooses the control value $\al_i := \al_i^*$. This value is independent of terminal time $T$, and of the values of $Y_j,\, j = 1, \ldots, N$. This corresponds to the classical solution of the Merton problem. If $r$ is constant (independent of $t$), then $\al_i^*$ and $h_i^*$ are also constant. Comparing~\eqref{eq:anzats} with the terminal condition, we have:
\begin{equation}
\label{eq:terminal}
g_{ij}(T) = \de_{ij} = 
\begin{cases}
1,\, i = j;\\
0,\, i \ne j,
\end{cases}
\ \mbox{for}\ j = 0, \ldots, N.
\end{equation}
Next, plug the anzats~\eqref{eq:anzats} into~\eqref{eq:HJB-private}. Note that all second-order  derivatives of the anzats~\eqref{eq:anzats} are equal to zero, and first-order derivatives are
\begin{equation}
\label{eq:space-derivative}
\frac{\pa\Phi_i}{\pa y_j} = g_{ij}(t),\ j = 1, \ldots, N. 
\end{equation}
In addition, the time derivative of this value function $\Phi$ from~\eqref{eq:anzats} is 
\begin{equation}
\label{eq:time-derivative}
\frac{\pa\Phi}{\pa t} = g'_{i0}(t) + \SL_{j=1}^Ng'_{ij}(t)y_j.
\end{equation}
Combining~\eqref{eq:alpha-max}, ~\eqref{eq:h-max}, ~\eqref{eq:space-derivative},~\eqref{eq:time-derivative}, we get that the HJB equation~\eqref{eq:HJB-private} takes the form
\begin{equation}
\label{eq:HJB-new}
g_{i0}'(t) + \SL_{j=1}^Ng_{ij}'(t)y_j + \SL_{j=1}^Nh_j^*(t)g_{ij}(t) + \frac1{N}\SL_{j=1}^N\SL_{k=1}^Nc_{jk}(t)(y_k - y_j)g_{ij}(t) = 0.
\end{equation}
Comparing coefficients in~\eqref{eq:HJB-new} at each $y_j$, we see that 
\begin{equation}
\label{eq:ODE-1}
g_{ij}'(t) + \frac1{N}\SL_{k=1}^Ng_{ik}(t)c_{jk}(t) - \frac1{N}\SL_{k=1}^Ng_{ik}(t)c_{kj}(t) = 0,\ \ j = 1,\ldots, N.
\end{equation}
The free terms in~\eqref{eq:HJB-new} sum up to 
\begin{equation}
\label{eq:ODE-2}
g_{i0}'(t) + \SL_{j=1}^Nh^*_{j}(t)g_{ij}(t) = 0.
\end{equation}
Together with terminal conditions~\eqref{eq:terminal}, this system~\eqref{eq:ODE-1} and~\eqref{eq:ODE-2} of $N+1$ linear ODEs has a unique solution $g_{i0}, \ldots, g_{iN}$. This solves the HJB equation. 

\smallskip

To complete the proof, let us do the verification argument. Take a bounded adapted control $\al_j = (\al_j(t),\, 0 \le t \le T)$ for each $j = 1, \ldots, N$. Apply It\^o's formula for $\Phi_i(t, Y(t))$:
\begin{align}
\label{eq:ito-verification-1}
\begin{split}
\mathrm{d}\Phi_i(t, Y(t)) & =  \biggl[\frac{\pa\Phi_i}{\pa t}(t, Y(t)) + \frac12\SL_{j=1}^N\SL_{k=1}^N(1 + \al_j(t))a_{jk}\frac{\pa^2\Phi_i}{\pa y_j\pa y_k}(t, Y(t)) \\ & + \SL_{j=1}^N\Bigl[h_j(\al_j(t), r(t)) + \frac{1}{N}\SL_{k=1}^Nc_{jk}(t)(Y_k(t) - Y_j(t))\Bigr]\frac{\pa\Phi_i}{\pa y_j}(t, Y(t))\biggr]\,\mathrm{d}t \\ & + \SL_{j=1}^N\frac{\pa\Phi_i}{\pa y_j}(t, Y(t))(1 + \al_j(t))\mathrm{d}W_j(t).
\end{split}
\end{align}
Using the boundedness of $\al_j = (\al_j(t),\, 0 \le t \le T)$, we get that the stochastic integral term in~\eqref{eq:ito-verification-1} has expectation zero. Combining~\eqref{eq:HJB-private} with~\eqref{eq:ito-verification-1}, we get that  
$(\Phi_i(t, Y(t)),\, t \ge 0)$ is a supermartingale for all admissible (adapted bounded) controls $\al_i$, but a martingale for the control $\al_i^*$. Recall that $\Phi_i(T, y) = y_i$. Therefore, $\mathbb E\,\Phi_i(0, Y(0)) \ge \mathbb E\,\Phi_i(T, Y(T)) = \mathbb E\, Y_i(T)$, with equality for the control $\al_i^*$. From here it immediately follows that $\al_i^*$ is indeed the optimal control.
\end{proof}

\subsection{The dynamics of banks under their optimal investment choices} 
Under the optimal control~\eqref{eq:alpha-max}, the processes $Y_i,\, i  =1, \ldots, N$ (we shall denote them by $Y_i^*$) satisfy the following system of stochastic differential equations: 
\begin{equation}
\label{eq:new-dynamics}
\md Y_i^*(t) = \md M_i^*(t) + \frac1{N}\left[\SL_{j=1}^Nc_{ij}(t)(Y^*_j(t) - Y^*_i(t))\right]\md t,\ i = 1, \ldots, N,
\end{equation}
where $M_1^*, \ldots, M_N^*$, are given by 
$$
\md M_i^*(t) = h_i(\al_i^*(t), r(t))\md t + \si_i(1 + \al_i^*(t))\,\md W_i(t).
$$
If $r = \const$, then $M^*$ is an $N$-dimensional Brownian motion with drift vector and covariance matrix given by
\begin{equation}
\label{eq:optimal-drift}
\mu^* = (\mu_1^*, \ldots,\mu_N^*),\ \ \mu_i^* := h_i(\al_i^*, r).
\end{equation}
\begin{equation}
\label{eq:optimal-cov}
A^* := (a^*_{ij})_{i, j = 1, \ldots, N} = \diag((1 + \al_i^*)^2,\, i = 1, \ldots, N)\,A.
\end{equation}
The dynamics~\eqref{eq:new-dynamics} is similar to that in \cite{JP-Rene}. If $r(t)$ does not depend on $t$, then $M^* = (M^*_1, \ldots, M^*_N)'$ , like $(M_1, \ldots, M_N)'$, is an $N$-dimensional Brownian motion, but with different drift vector and covariance matrix. As in~\eqref{eq:mean}, we define
$$
\ol{Y}^*(t) = \frac1N\SL_{i=1}^NY^*_i(t).
$$
Averaging equations~\eqref{eq:new-dynamics} and using the symmetry property $c_{ij} = c_{ji}$, we have:
\begin{equation}
\label{eq:mean-Y}
\ol{Y}^*(t) = \frac1N\SL_{i=1}^NM_i^*(t),
\end{equation}
The interest rate $r$ controls the overall size of the system, measured by $\ol{Y}$. Express~\eqref{eq:mean-Y} as:
\begin{equation}
\label{eq:average-Y}
\md\ol{Y}^*(t) = g(r(t))\,\md t + \rho(r(t))\,\md \ol{W}(t),
\end{equation}
where $\ol{W}$ is a standard Brownian motion, and the coefficients $g(\cdot)$ and $\rho(\cdot)$ are defined as:
\begin{equation}
\label{eq:new-drift-g}
g(r) := \frac1N\SL_{i=1}^Ng_i(r),\ \ g_i(r) := 
\begin{cases}
\frac{(\mu_i - r)^2}{2\si_i^2} + r,\ \ r \le \mu_i - \si_i^2;\\
\mu_i - \frac{\si_i^2}2,\ \ r \ge \mu_i - \si_i^2;\\
\frac{\mu_i^2}{2\si_i^2},\ \ \mu_i < \si_i^2.
\end{cases}
\end{equation}
\begin{equation}
\label{eq:new-diffusion}
\rho^2(r) := \frac1{N^2}\SL_{i=1}^N\SL_{j=1}^Na_{ij}\rho_i(r)\rho_j(r),\ \ \rho_i(r) := 
\begin{cases}
\frac{\mu_i - r}{\si_i^2},\ \ 0 \le r \le \mu_i - \si_i^2;\\
1,\ \ 0 \le \mu_i - \si_i^2 \le r;\\
\frac{\mu_i}{\si_i^2},\ \mu_i \le \si_i^2.
\end{cases}
\end{equation}

    \begin{figure*}[htbp]
        \centering
        \subfloat[$r=0$]{\includegraphics[width = 4.5cm]{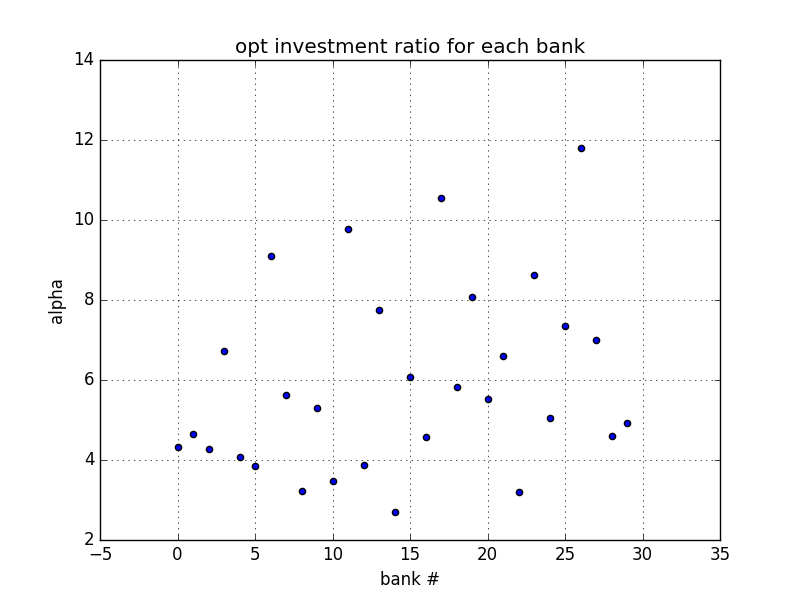}}\label{c0_r0_corrl0}
        \subfloat[$r=0.12$]{\includegraphics[width = 4.5cm]{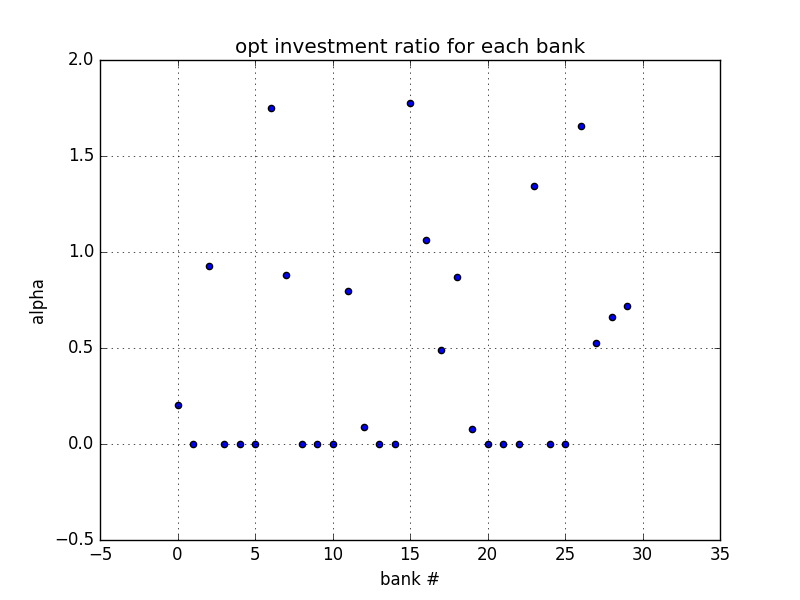}}\label{c0_r12_corrl0}
        \subfloat[$r=0.20$]{\includegraphics[width = 4.5cm]{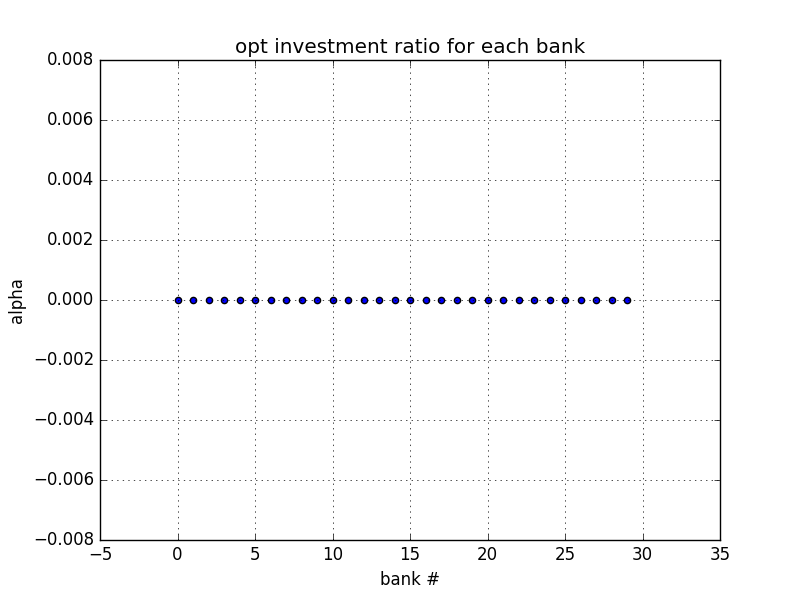}}\label{c0_r20_corrl0}
          \caption[ ]
        {\small \textit{We use the following parameters for the simulations: $N=30$ bank, time horizon $T =1$, no correlation $\rho_{0} = 0$, no interbank flows $c_{i,j}=0$, 1000 time steps and $\mu_i, \si_i,\, i = 1, \ldots, N$ i.i.d. uniform $[0.1, 0.2]$  } } 
        \label{c0_corrl0}
    \end{figure*}

To illustrate the optimal choice of the investment ratio $\al_i^* = \al_i,\, i = 1, \ldots, N$, we made some simulations. Take $N = 30$ banks, with $\mu_i, \si_i,\, i = 1, \ldots, N$ i.i.d. uniform $[0.1, 0.2]$. Then $\mu_i \ge \si_i^2$ for all $i$; that is, all portfolios are profitable to invest, at least for zero interest rate $r = 0$. First, in Figure~\ref{c0_corrl0} we assume~\eqref{eq:diagonal-A}, that is, the portfolio processes $S_1, \ldots, S_N$, are independent. We also assume that there are no flows: 
$$
c_{ij}(t) \equiv 0,\ i, j = 1, \ldots, N.
$$
We take three interest rates $r$: $0\%,\,12\%$, and $20\%$ respectively. As expected, increasing the interest rate forces the banks to borrow less and thus optimal investment ratio $\alpha^*$ becomes 0 in Figure~\ref{c0_corrl0}(C) while it varied between $2$ to $12$ in Figure~\ref{c0_corrl0}(A). 

\smallskip

    \begin{figure*}[htbp]
        \centering
        \subfloat[All the banks]{\includegraphics[width=4cm]{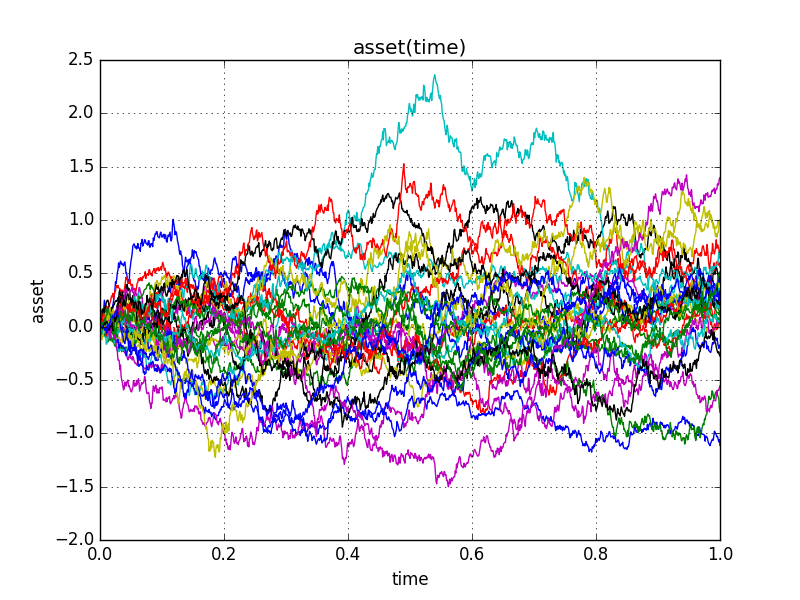}}
        \subfloat[$i = 1, \ldots, 10$]{\includegraphics[width=4cm]{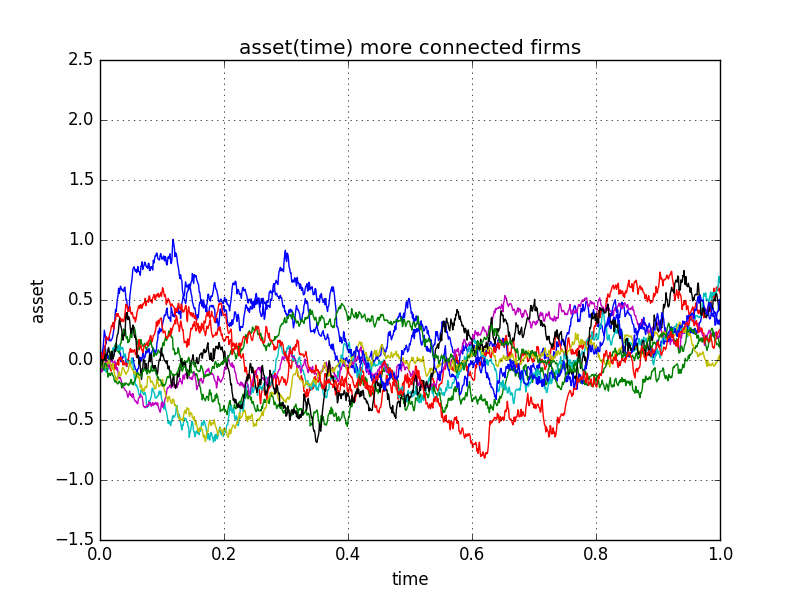}}
        \subfloat[$j = 11, \ldots, 30$]{\includegraphics[width=4cm]{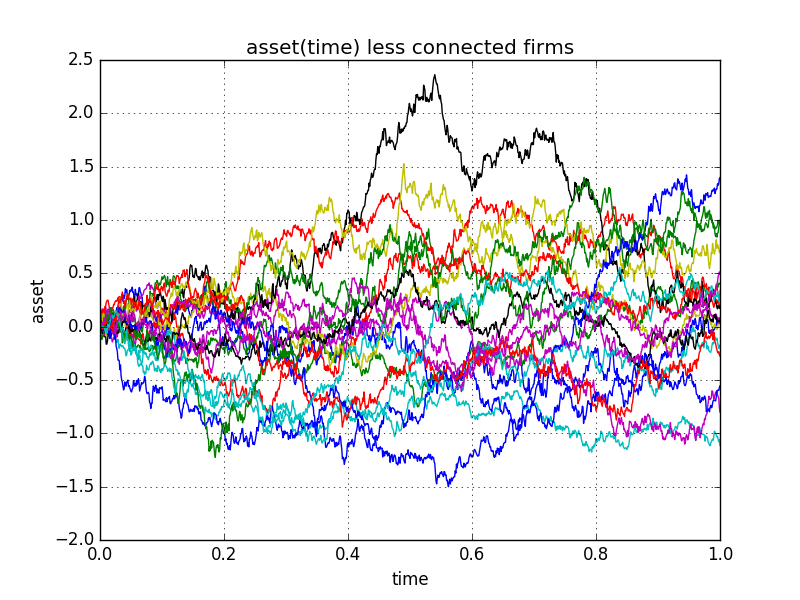}}
        \caption[ ]
        {\small \textit{Evolution of the logarithmic capital $Y_i(t)$ of banks, $i = 1, \ldots, N$. We use the following parameters: interest rate $r=0$, $N=30$ banks, time horizon $T =1$, no correlation: $\rho_{0} = 0$, interbank flows $c_{i,j}$ are as in equation~\eqref{eq:example-flows} , $1000$ time steps, and $\mu_i, \si_i,\, i = 1, \ldots, N$ are i.i.d. uniform on $[0.1, 0.2]$  } } 
        \label{nonZeroflows_r0_rho0}
    \end{figure*}

    \begin{figure*}[htbp]
        \centering
         \subfloat[All the banks]{\includegraphics[width=4cm]{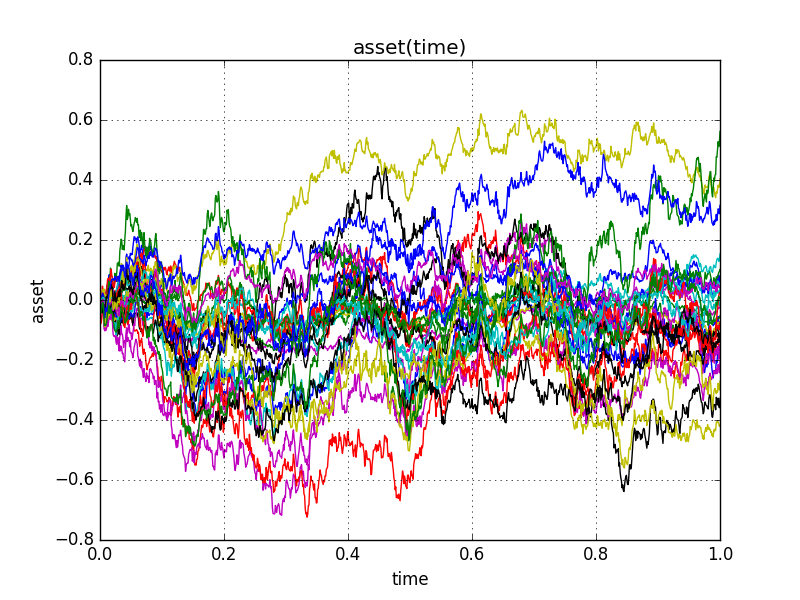}}
         \subfloat[$i = 1, \ldots, 10$]{\includegraphics[width=4cm]{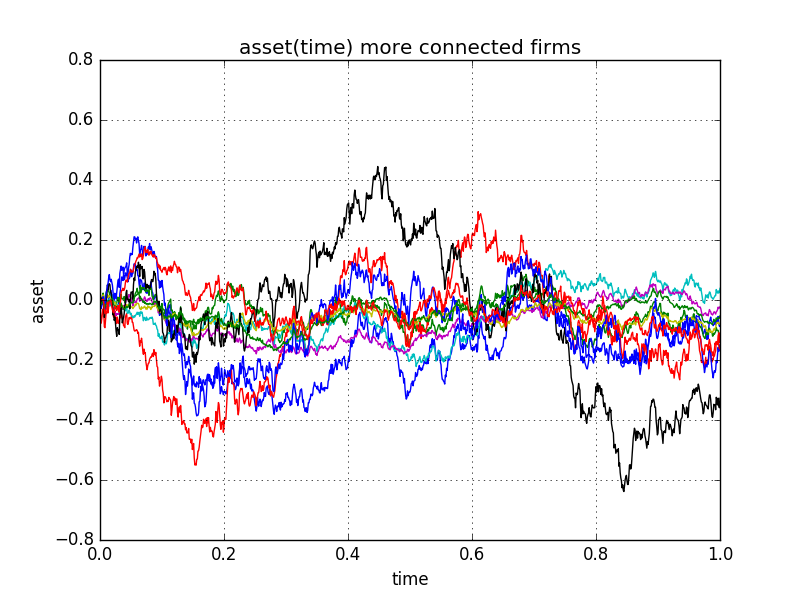}}
                  \subfloat[$i = 11, \ldots, 30$]{\includegraphics[width=4cm]{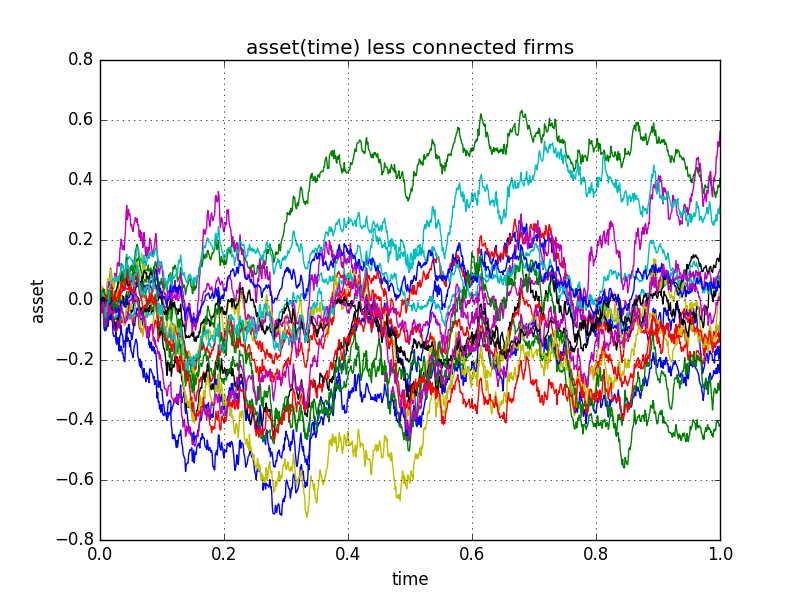}}
        \caption[ ]
        {\small \textit{Evolution of the logarithmic capital of banks $Y_i(t)$.We use the following parameters: interest rate $r=8\%$, $N=30$ banks, time horizon $T =1$, correlation coefficient $\rho_{0} = 0.5$, interbank flows $c_{i,j}$ as in equation~\eqref{eq:example-flows} , $1000$ time steps and $\mu_i, \si_i,\, i = 1, \ldots, N$ i.i.d. uniform $[0.1, 0.2]$} } 
        \label{nonZeroflows_r08_rho50}
    \end{figure*}
    
    Next, in Figure~\ref{nonZeroflows_r0_rho0} we assume that the portfolio processes are independent, as in~\eqref{eq:diagonal-A}, but there are flows: 
\begin{equation}
\label{eq:example-flows}
c_{ij} = 
\begin{cases}
10,\ i, j = 1, \ldots, 10;\\
0.5,\ \mbox{else}.
\end{cases}
\end{equation}
We observe that the banks with significant flows tend to have wealth dynamics which are more ``tied" together. Moreover, this adds to the stability to the system as we observe lesser defaults for $i=1,\ldots,10$ compared to $j=11,\ldots,30$.
    
    \smallskip

Finally, in Figure~\ref{nonZeroflows_r08_rho50} we assume that the portfolio processes are correlated, as in~\eqref{eq:correlated-assets}, with $\rho_0 := 0.5$, flows are given by~\eqref{eq:example-flows} and interest rate $r$ is $8\%$. Compared to Figure~\ref{nonZeroflows_r0_rho0}, the impact of correlation on the dynamics of the banks is clearly visible through movement of the wealth  dynamics strongly tied together.  

\subsection{Systemic risk} Much of current research is devoted to {\it systemic risk}, that is, the probability of multiple bank defaults, and propagation of defaults through the system (in other words, contagion). To illustrate the probability of default of banks under different scenarios, we present the histogram and the empirical cumulative distribution function for number of defaults with $N=100$ banks and $1000$ simulations. We assume the default threshold $D=-1$ in logarithmic wealth. That is, firms default if $Y_i(t)<D$ for some $t \in [0,T]$. Denote the (random) number $\mathcal D$ of defaults:
\begin{equation}
\label{eq:number-of-defaults}
\mathcal D := \SL_{i=1}^N1\left(\min\limits_{0 \le t \le T}Y_i(t) < D\right),
\end{equation}
 First, in Figure~\ref{histogram_noflows_corrl0} we assume no interbank flows and  independent portfolio process under different interest rate scenarios:

    \begin{figure*}[htbp]
        \centering
        \subfloat[$r=0$]{\includegraphics[width=4cm]{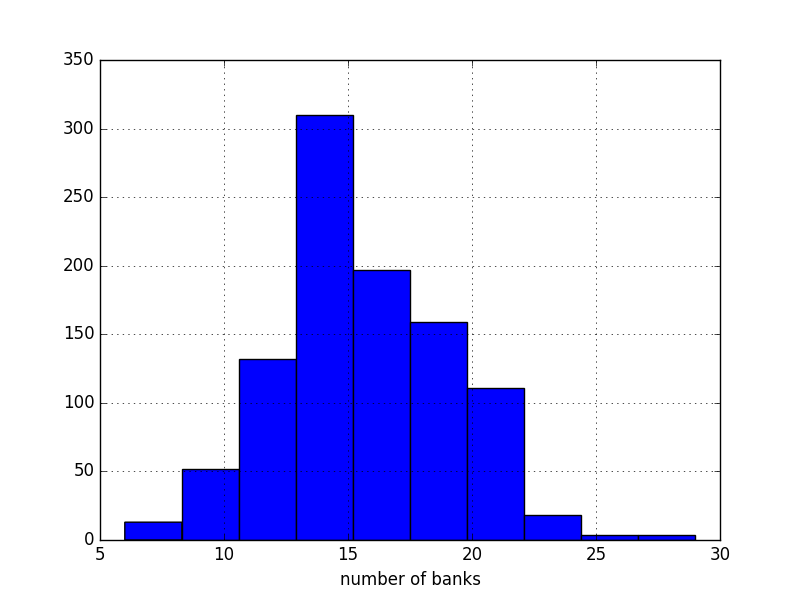}}
        \subfloat[$r=0.05$]{\includegraphics[width=4cm]{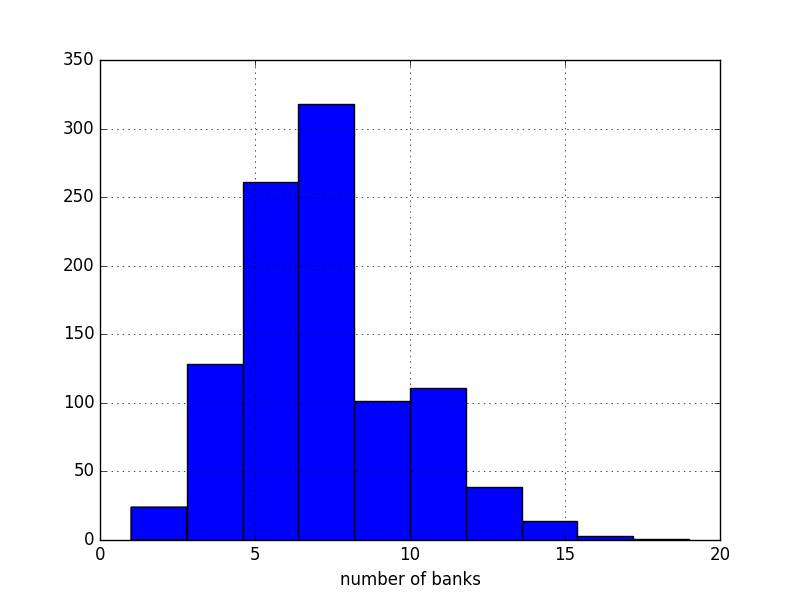}}
        \subfloat[$r=0.08$]{\includegraphics[width=4cm]{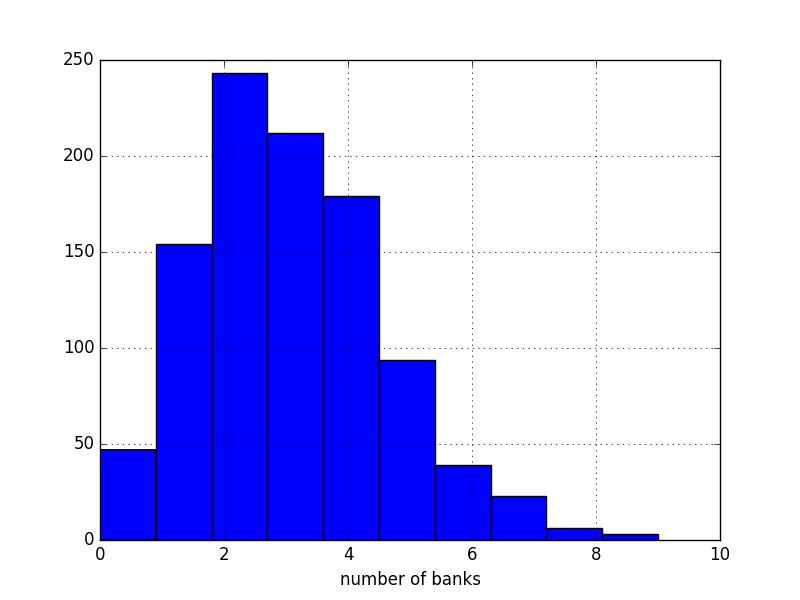}}
        \caption[ ]
        {\small \textit{Number of banks in default, whose log capital $Y_i(t)$ at some time $t \in [0,T]$ goes below $D = -1$. We use the following parameters: $N=100$ banks, $1000$ simulations, no correlation: $a_{ij}=\sigma_{i}^2\delta_{ij}$, no interbank flows: $c_{ij}=0$ for $i, j = 1, \ldots, N$, $100$ time steps, and $\mu_i, \si_i,\, i = 1, \ldots, N$ are i.i.d. uniform on $[0.1, 0.2]$ } } 
        \label{histogram_noflows_corrl0}
    \end{figure*}

    \begin{figure*}[htbp]
        \centering
            \includegraphics[width=0.4\textwidth]{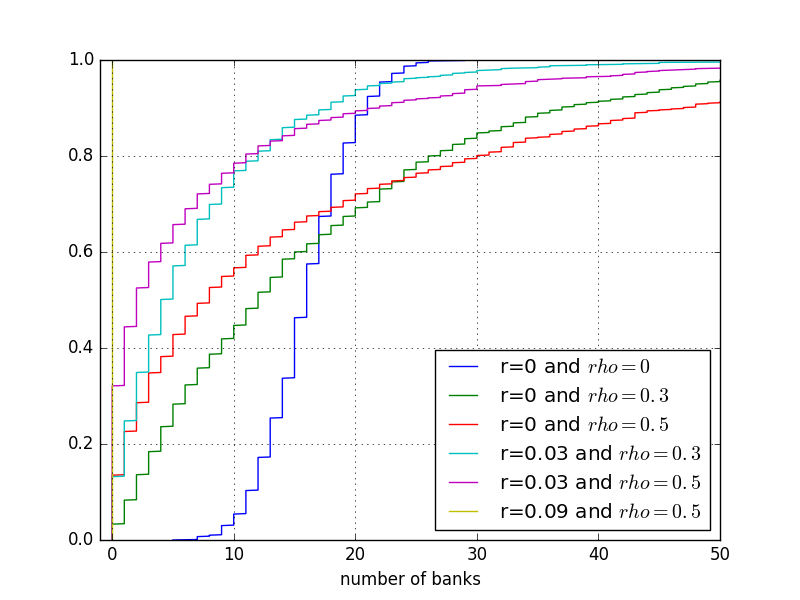}
            \caption[]%
            {\small \textit{ Empirical CDF of $\mathfrak D$, the number of banks in default, with $N=100$ banks, $1000$ simulations, $\mu_i=\sigma_i=0.1$ for $i=1,\ldots,N$}}    
            \label{ecdf_c0}
    \end{figure*}

\smallskip

Next, in Figure~\ref{ecdf_c0} we present empirical cumulative distribution function (CDF) for the number of defaulted banks assuming correlated portfolios and no interbank flows for different interest rates. The corresponding histogram is presented in figure \ref{hist_corrlAssets_noflows}. As stated in previous studies, increase in correlation increases the probability of large defaults and at the same time reducing the small default probabilities, similar to flocking behaviour in various biological studies. So, in Figure~ \ref{defaultProb} we present empirical estimates of $\MP(\mathcal D > 60)$ and $\MP(\mathcal D < 5)$, for $N = 100$ banks, as a function of the correlation between their portfolio process at different interest rates. As expected, we observe the increase in probability of large and small default as the correlation increases. However, increasing the interest rate reduces the probability of large default at the expense of small default probability. Finally, in Figure~\ref{ecdf_cconstant} we present the empirical CDF for the number of defaulted banks assuming correlated portfolio process and constant interbank flows $c_{ij}=a$ for $i, j = 1, \ldots, N$, where $a \in \{0,0.5,1 \}$. We observe that interbank flows help stabilize the system and reduce the probability of default.

    \begin{figure*}[htbp]
        \centering
        \subfloat[$r=0$ and $\rho_0=0$]{\includegraphics[width=6.5cm]{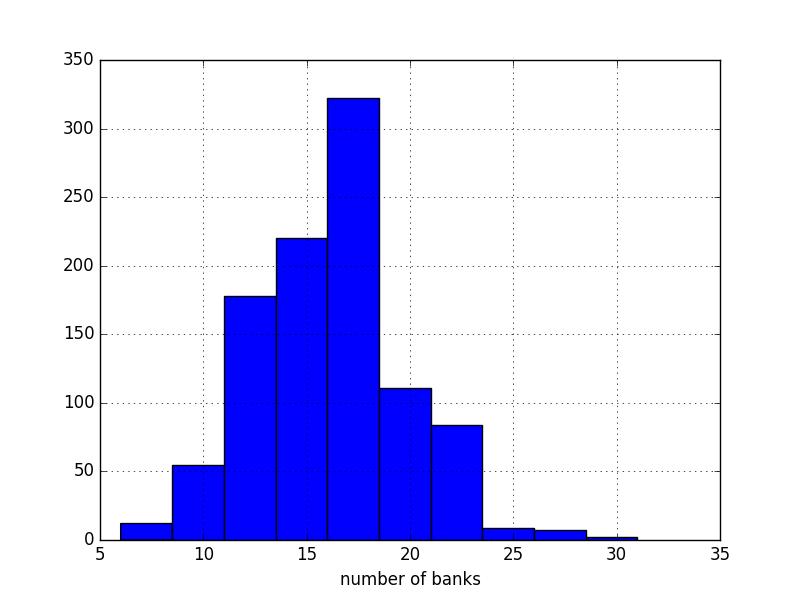}}
        \subfloat[$r=0$ and $\rho_0=0.5$]{\includegraphics[width=6.5cm]{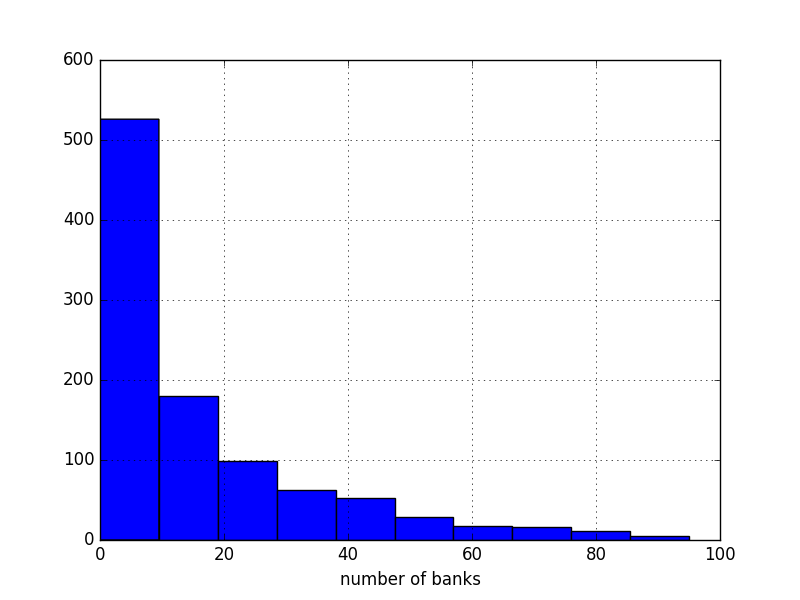}}
                \vskip\baselineskip
        \subfloat[$r=0.03$ and $\rho_0=0.3$]{\includegraphics[width=6.5cm]{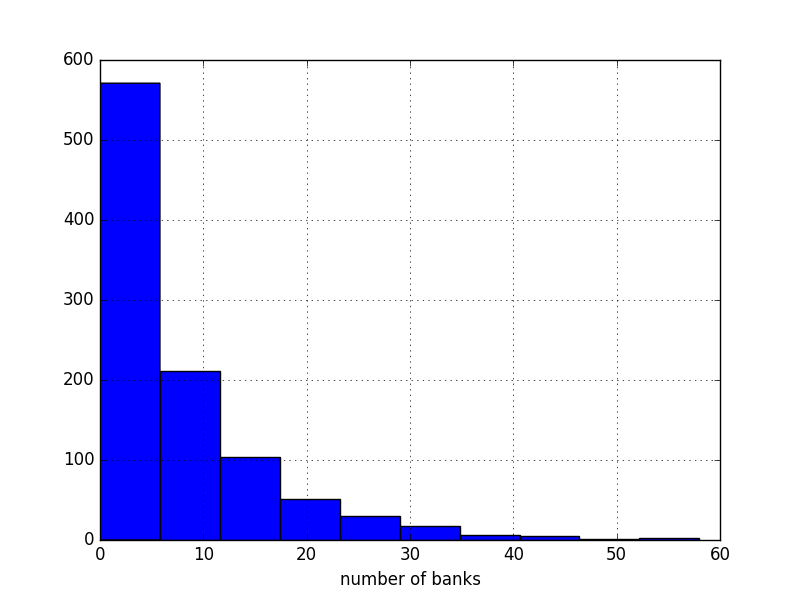}}
                \subfloat[$r=0.05$ and $\rho_0=0.3$]{\includegraphics[width=6.5cm]{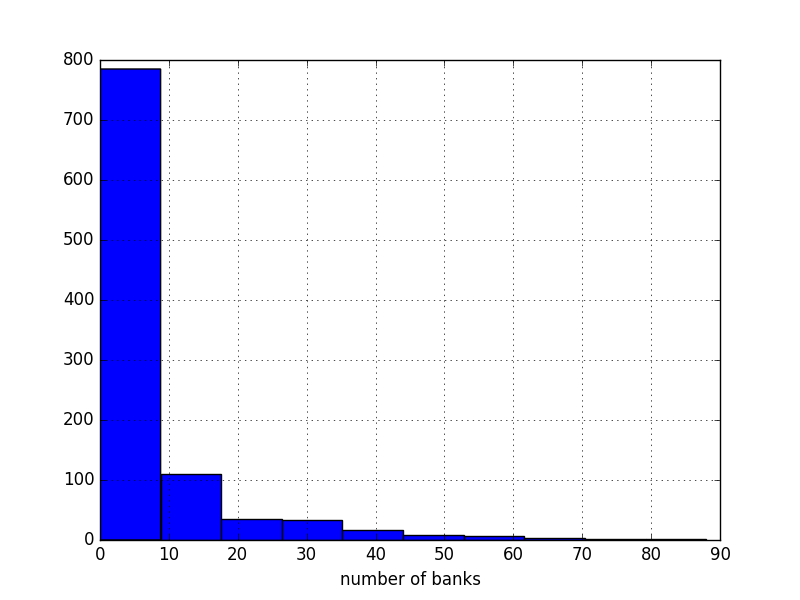}}
        \caption[ ]
        {\small \textit{Histogram of the number of banks defaulting. We use the following parameters: $N=100$ banks, $1000$ simulations, $\mu_i=\sigma_i=0.1$ for $i=1,\ldots,N$, and no interbank flows: $c_{ij}=0$ for $i, j = 1, \ldots, N$} } 
        \label{hist_corrlAssets_noflows}
    \end{figure*}

    \begin{figure*}[htbp]
        \centering
        \subfloat[Probability of a large default: $\mathcal D > 60$]{\includegraphics[width=6.5cm]{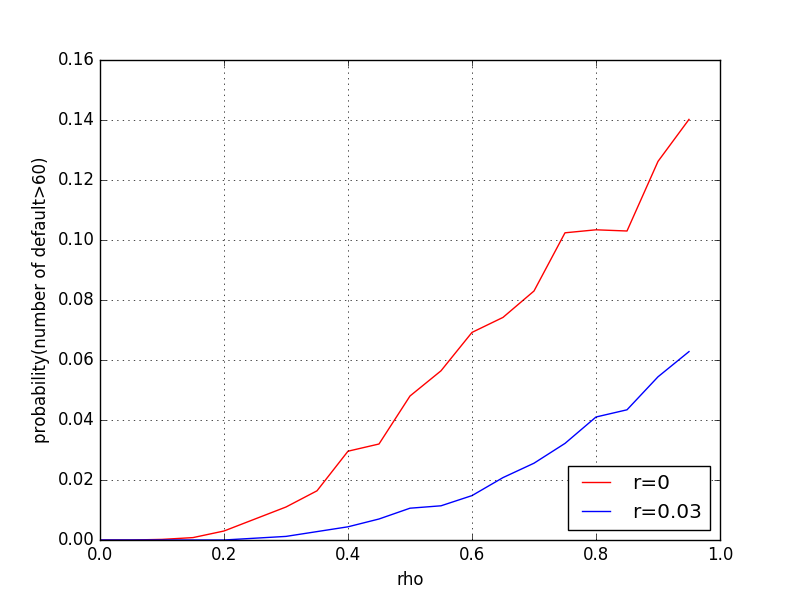}}
        \subfloat[Probability of a small default: $\mathcal D < 5$]{\includegraphics[width=6.5cm]{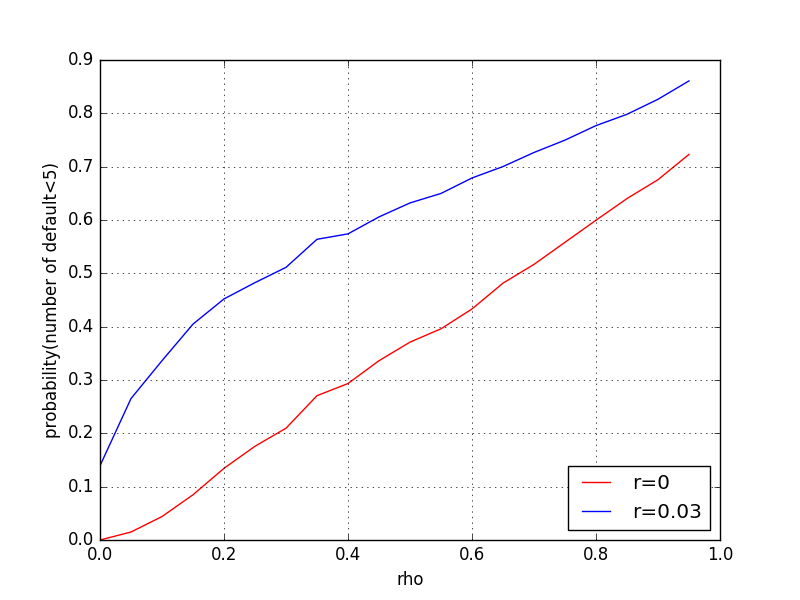}}
        \caption[]%
        {\small \textit{Empirical estimates of the probabilities of large and small defaults: $\MP(\mathcal D > 60)$ and $\MP(\mathcal D < 5)$, respectively, as a function of correlation between portfolio process $\rho_0$ at different interest rates. We use the following parameters: $N=100$ banks, $5000$ simulations, $\mu_i=\sigma_i=0.1$ for $i=1,\ldots,N$, and interbank flow rates $c_{ij}=0$ for $i, j = 1, \ldots, N$}}    
         \label{defaultProb}
    \end{figure*}

    \begin{figure*}[htbp]
        \centering
        \subfloat[$\rho_0=0.5$ and $r=0$]{\includegraphics[width=6.5cm]{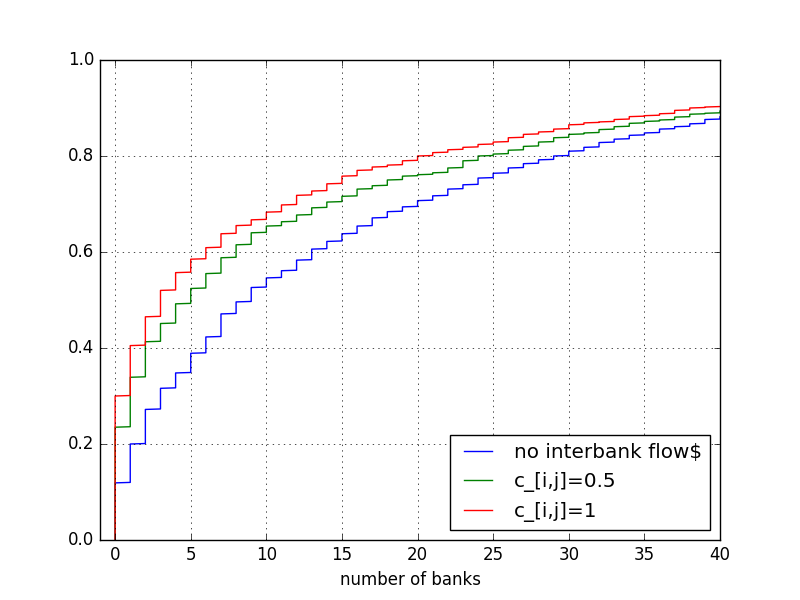}}
        \subfloat[$\rho_0=0.5$ and $r=0.03$]{\includegraphics[width=6.5cm]{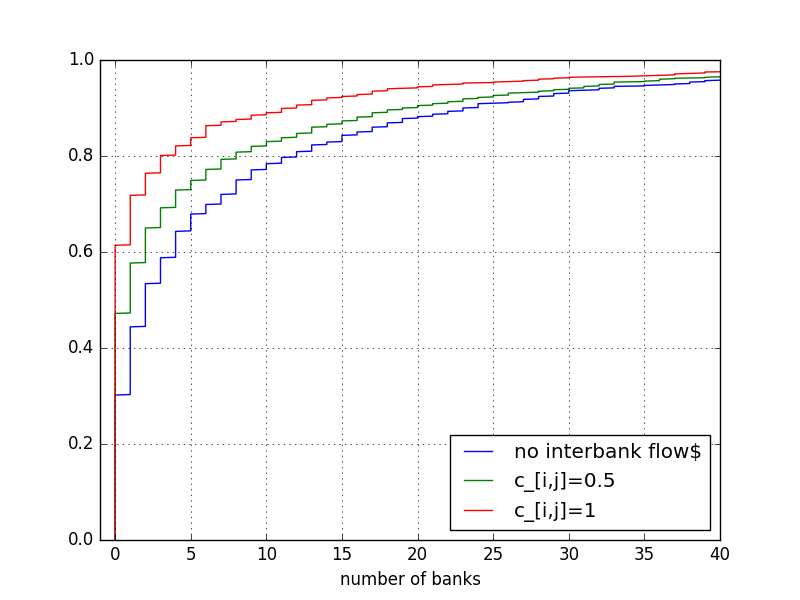}}
        \caption[]%
        {\small \textit{Empirical CDF of the number $\mathcal D$ of banks in default. We use the following parameters: $N=100$ banks, $1000$ simulations, $\mu_i=\sigma_i=0.1$ for $i=1,\ldots,N$, and constant interbank flows $c_{ij}=a$ for $i, j = 1, \ldots, N$, where $a \in \{0,0.5,1 \}$} }         
         \label{ecdf_cconstant}
    \end{figure*}

\section{Optimal central bank policy} In this section, we assume that the central bank has to choose the interest rate $r$ in an optimal way, so that, after banks make their choice as in the previous section, optimal policy choice is achieved. We assume that banks make optimal (for them) choices and we omit all asterisks from notation of processes. This can be thought of as a principal's problem within the principal-agent problem framework. Let us now revisit the description of policymaking by the central bank. 

\smallskip

Its tool is the {\it interest rate} $r$, which the central bank uses to control the overall amount of capital in the system, measured by the $\ol{Y}^*$ from~\eqref{eq:average-Y}. If the interest rate is low, the growth rate $g(r)$ from~\eqref{eq:new-drift-g} and the volatility $\rho^2(r)$ from~\eqref{eq:new-diffusion} are large. A more risk-averse central bank can choose therefore a larger $r$. One can apply a concave utility function to $\ol{Y}^*(t)$,  and solve the stochastic control problem for this $r$. In this section, we apply the exponential (CARA: constant relative risk aversion) utility function to $\ol{Y}^*(t)$. 

\smallskip

The private banks wish to maximize their expected logarithmic net worth $Y_i(t)$. In other words, they have logarithmic utility $U(x) := \log x$, which shows their aversion to risk. However, {\it in terms of logarithmic capital}, their utility function is linear. Now, if the central bank was as risk-averse as private banks, she would also try to maximize 
$$
\ME(Y_1(T) + \ldots + Y_N(T)),\ \ \mbox{or, alternatively,}\ \  \ME\ol{Y}(T),
$$ 
for a time horizon $T > 0$. Below, we show that the central bank would then choose zero interest rate $r = 0$, because this would produce the same result as the private banks were aiming for. 

\smallskip

Now, suppose the central bank is even more risk-averse than private banks. This should manifest itself in the utility function being concave (rather than linear) even in logarithmic terms. Consider, for example, a commonly used exponential (CARA) utility function:
\begin{equation}
\label{eq:CARA-1}
U_{\la}(y) := -e^{-\la y}.
\end{equation}
Assume the central bank maximizes expected terminal utility:
\begin{equation}
\label{eq:central}
 \sup\limits_r\ME U_{\la}(\ol{Y}(T)),
\end{equation}
where the supremum in~\eqref{eq:central} is chosen over all bounded adapted controls $r$. We can alternatively choose instead of~\eqref{eq:CARA-1} the utility function as
\begin{equation}
\label{eq:CARA-2}
U_{\la}(y) = \frac1{\la}\left(1 - e^{-\la y}\right).
\end{equation}
There is no difference between~\eqref{eq:CARA-1} and~\eqref{eq:CARA-2} when we try to maximize~\eqref{eq:central}, but writing~\eqref{eq:CARA-2} highlights the risk-aversion of the central bank. As $\la \downarrow 0$, the function $U_{\la}$ from~\eqref{eq:CARA-2} satisfies:
$$
U_{\la}(y) \to y.
$$
The commonly used {\it absolute risk aversion} is calculated for~\eqref{eq:CARA-2} as follows:
$$
-\frac{U''_{\la}(y)}{U'_{\la}(y)} = \la.
$$
In other words, $\la > 0$ is the coefficient of risk aversion (of the central bank relative to private banks). For $\la = 0$ the central bank is not risk-averse at all (at least not more  than private banks). 

\smallskip

\begin{theorem} An optimal interest rate $r(t)$ for the problem~\eqref{eq:central} is given by a constant $r = r^*$ which maximizes the following expression:
\begin{equation}
\label{eq:optimal-r}
w(r,\lambda) :=  g(r) - \frac{\la}2\rho^2(r). 
\end{equation}
\label{thm:optimal-r}
\end{theorem}

\begin{remark}
It is interesting to note that the optimal interest rate $r$ does not depend on the flow rates $c_{ij}$. This is because we measure the size of the system by the stochastic process $\ol{Y}(t)$. This process satisfies a stochastic differential equation with coefficients independent of $c_{ij}$. These coefficients do depend on the optimal controls $\al_i^*$. However, as we mentioned in Remark~\ref{remark:independent-of-flows}, these optimal controls $\al_i^*$, in turn, do not depend on the flow rates, because of our special choice of logarithmic utility function. 
\end{remark}

\begin{proof} The HJB equation for the function 
$$
\Phi(t, y) := \sup\limits_r\ME\left[U_{\la}(\ol{Y}(T))\mid \ol{Y}(t) = y\right]
$$
where the supremum is taken over all bounded adapted controls $r = (r(t),\, 0 \le t \le T)$, takes the form 
\begin{equation}
\label{eq:HJB}
\frac{\pa\Phi}{\pa t} + \sup\limits_{r \ge 0}\left[\frac{\pa\Phi}{\pa y}g(r) + \frac12\frac{\pa^2\Phi}{\pa y^2}\rho^2(r)\right] = 0,
\end{equation}
with terminal condition $\Phi(T, y) = U_{\la}(y)$. Try the following form: 
\begin{equation}
\label{eq:product-form}
\Phi(t, y) = f(t)U_{\la}(y).
\end{equation}
From~\eqref{eq:product-form}, we can calculate derivatives with respect to $t$ and $y$:
\begin{equation}
\label{eq:derivatives-1}
\frac{\pa\Phi}{\pa t} = f'(t)U_{\la}(y),\ \frac{\pa\Phi}{\pa y} = -\la \Phi,\ \frac{\pa^2\Phi}{\pa y^2} = \la^2\Phi.
\end{equation}
Plug~\eqref{eq:derivatives-1} into~\eqref{eq:HJB}. Because $\Phi < 0$, we can rewrite~\eqref{eq:HJB} as
$$
f'(t) + f(t)\cdot\inf\limits_{r \ge 0}\left[-\la g(r) + \frac{\la^2}2\rho^2(r)\right] = 0.
$$
This, in turn, is equivalent to 
\begin{equation}
\label{eq:transformed-HJB}
\frac{f'(t)}{\la f(t)} = \sup\limits_{r \ge 0}\left[g(r) - \frac{\la}2\rho^2(r)\right] =: k_0.
\end{equation}
Since we have $\Phi(T, y) < 0$ and $U_{\la}(y) < 0$, for compatibility we need to show that $f(t) > 0$ for all $t$. From the terminal condition $\Phi(T, y) = U_{\la}(y)$ combined with~\eqref{eq:product-form}, we have: $f(T) = 1$. The equation~\eqref{eq:transformed-HJB} can be written as $f'(t) = \la k_0f(t)$, which gives us $f(t) = \exp\left(\la k_0(t-T)\right)$. Therefore, $f(t)$ is positive. 

\smallskip

Finally, let us do the verification argument to complete the proof. The idea is similar to the verification argument in Theorem \ref{thm:agent}. Assume $r_* = (r_*(t),\, 0 \le t \le T)$ is our constant control from~\eqref{eq:optimal-r}, found from~\eqref{eq:HJB}, and $r = (r(t),\, 0 \le t \le T)$ is some other admissible (adapted bounded) control. Apply the function $\Phi(t, \cdot)$ to the process $\ol{Y}$. By It\^o's formula, 
\begin{align}
\label{eq:ito-second-verification}
\begin{split}
\mathrm{d}\Phi(t, \ol{Y}(t))  = \biggl[\frac{\pa\Phi}{\pa t}(t, \overline{Y}(t)) & + \frac{\pa\Phi}{\pa y}(t, \overline{Y}(t))g(r(t)) + \frac12\frac{\pa^2\Phi}{\pa y^2}(t, \overline{Y}(t))\rho^2(r(t))\biggr]\,\mathrm{d}t \\ & + \frac{\pa\Phi}{\pa y}(t, \overline{Y}(t))\rho(r(t))\,\mathrm{d}\overline{W}(t).
\end{split}
\end{align}
Comparing~\eqref{eq:HJB} with~\eqref{eq:ito-second-verification}, we get that $\Phi(t, \ol{Y}(t))$ is a supermartingale for the control $r$, but a martingale for the control $r_*$. Indeed, by boundedness of $r(t)$, the expectation of the stochastic integral in~\eqref{eq:ito-second-verification} is zero. Since $\Phi(T, y) = U_{\la}(y)$, we get: $\mathbb E\, U_{\la}(\ol{Y}(T)) = \mathbb E\, \Phi(T, \ol{Y}(T)) \le \mathbb E\, \Phi(0, \ol{Y}(0)) $, with equality for the control $r_*$. The result immediately follows from here. 
\end{proof}

Let us find the $r$ which corresponds to the maximum in the right-hand side of~\eqref{eq:transformed-HJB}. This depends on the structure of the vector $g$ and the matrix $A$.

\smallskip

If $\mu_i \le \si_i^2$ for all $i = 1, \ldots, N$, then all investments are too unprofitable to borrow money for them. Then the interest rate policy cannot influence the behavior of private banks. This corresponds to the case of the {\it liquidity trap}, when conventional monetary policy no longer works. From now on until the end of this section, let us assume that all investments are attractive:
$$
\mu_i \ge \si_i^2,\ \ i = 1, \ldots, N.
$$

\smallskip

{\bf (3.a)} Assume $S_1 = \ldots = S_N$: all investments are the same. Then we have:
$$
g_1 = \ldots = g_n =: g,\ \mbox{and}\ \si_1 = \ldots = \si_N =: \si;
$$
$$
g(r) - \frac{\la}2\rho^2(r) = 
\begin{cases}
\frac{(\mu - r)^2}{2\si^2}(1 - \la),\ r \le \mu - \si^2;\\
\mu - \frac{\si^2}2(1 + \la),\  r \ge \mu - \si^2.
\end{cases}
$$
The maximum is attained at $r = 0$ for 
$$
\la < \la_* := 1 - 2\left(\frac{\mu}{\si^2} + 1\right)^{-1},
$$
and at any $r \ge \mu - \si^2$ for $\la > \la_*$. This has the following meaning: the case $\la < \la_*$ corresponds to less risk-averse central bank, and in order to increase the total quantity of capital in the system, it wishes to slash the interest rate to zero. For the case $\la > \la_*$, however, the central bank is very risk-averse, and it increases the interest rate to prevent excessive borrowing and overheating of the financial system. 

\smallskip

{\bf (3.b)} Independent portfolio process: $a_{ij} = \si_i^2\de_{ij}$. Then 
$$
g(r) - \frac{\la}2\rho^2(r) = \frac1N\SL_{i=1}^N\left[g_i(r) - \frac{\la}{2N}\si_i^2\rho_i^2(r)\right].
$$
This function attains maximum:
$$
\mbox{at}\ \ r = 0\ \ \mbox{for}\ \ \la < \la_{\min} := N\min\limits_{i = 1, \ldots, N}\left[1 - 2\left(\frac{\mu_i}{\si^2_i} + 1\right)^{-1}\right],
$$
$$
\mbox{at}\ \ r = \max\limits_{i = 1, \ldots, N}\left[\mu_i - \si_i^2\right]\ \ \mbox{for}\ \ \la > \la_{\max} := N\max\limits_{i = 1, \ldots, N}\left[1 - 2\left(\frac{\mu_i}{\si^2_i} + 1\right)^{-1}\right].
$$
In the general case, we do not have an explicit form for the optimal $r$ in case $\la \in [\la_{\min}, \la_{\max}]$. If $\mu_1 = \ldots = \mu_N = \mu$ and $\si_1 = \ldots = \si_N = \si$, we have $\la_{\min} = \la_{\max}$. Note that here the central bank chooses expansionary monetary policy (zero interest rate $r = 0$) for larger values of $\la$ than in case (3.b). This has the following interpretation: If the portfolios of banks are independent, then this creates diversification in the system and reduces risk. Therefore, even a relatively risk-averse central bank (large $\la$) can pursue aggressive expansionary monetary policy. 

\smallskip

{\bf (3.c)} Correlated portfolio process with same growth rates $\mu = \mu_i$ and volatilities $\si^2 = \si_i^2$. Assume the driving Brownian motions of these portfolio process are correlated as in~\eqref{eq:correlated-assets}.
After calculation , we get:
$$
g(r) = 
\begin{cases}
\frac{(\mu - r)^2}{\si^2} + r,\ r \le \mu - \si^2;\\
\mu - \frac{\si^2}2,\ r \ge \mu - \si^2;
\end{cases}
$$
$$
\rho^2(t) :=c\left(\frac{\mu - r}{\si}\wedge\si\right)^2,\ \ c := \lambda\left(\frac{N-1}N\rho_0 + \frac1N\right).
$$
Then we can find optimal $r$: this is 
$$
r^* = 
\begin{cases}
0,\ c < 1 - 2\left(\frac{\mu}{\si^2} - 1\right)^{-1};\\
\mu - \si^2,\ c > 1 - 2\left(\frac{\mu}{\si^2} - 1\right)^{-1}.
\end{cases}
$$
Note that for $\rho_0 = 1$ we get case (3.a), and for $\rho_0 = 0$ we get case (3.b). Case (3.c) is intermediate: there is diversification between portfolios of private banks, but this diversification is not complete. Therefore, a risk-averse central bank can pursue more expansionary monetary policy than in Case (3.a), but less so than in Case (3.b).

    \begin{figure*}[htbp]
        \centering
        \subfloat[optimal interest rate]{\includegraphics[width=6.5cm]{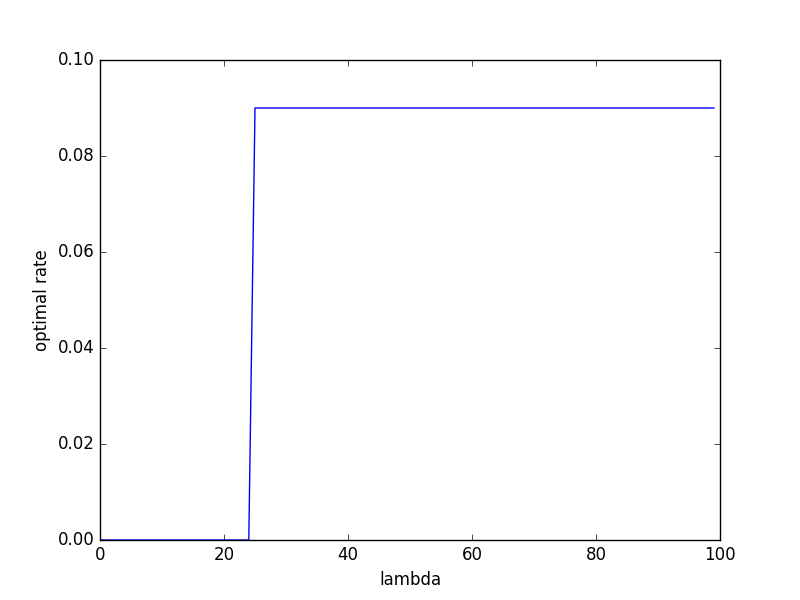}}
        \subfloat[$w(r, \lambda)$]{\includegraphics[width=6.5cm]{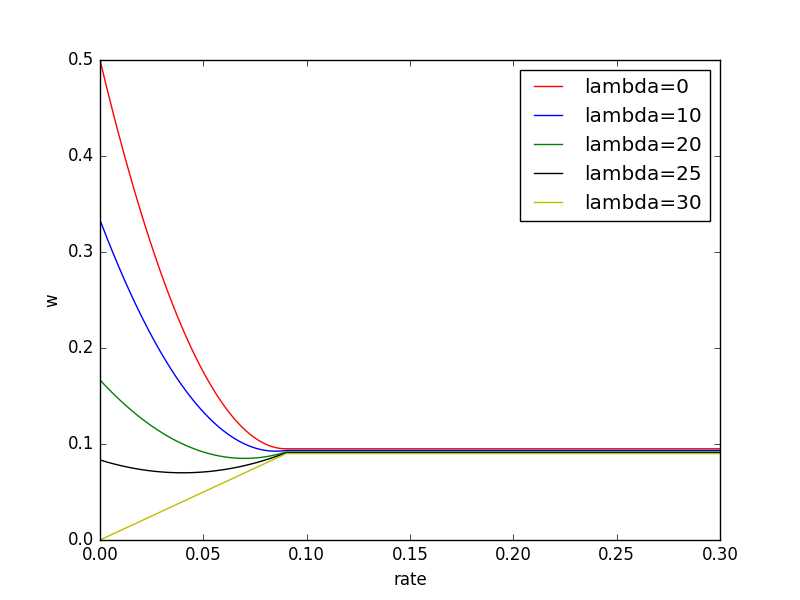}}
        \caption[ ]
        {\small \textit{Optimal interest rate with $N=30$ uncorrelated portfolio process: $\rho_0 = 0$, with $\mu_i=\sigma_i=0.1$ for $i = 1,\ldots,N$ } } 
        \label{optR_fixedMu_sigma}
    \end{figure*}

        \begin{figure*}[htbp]
        \centering
        \subfloat[optimal interest rate]{\includegraphics[width=4cm]{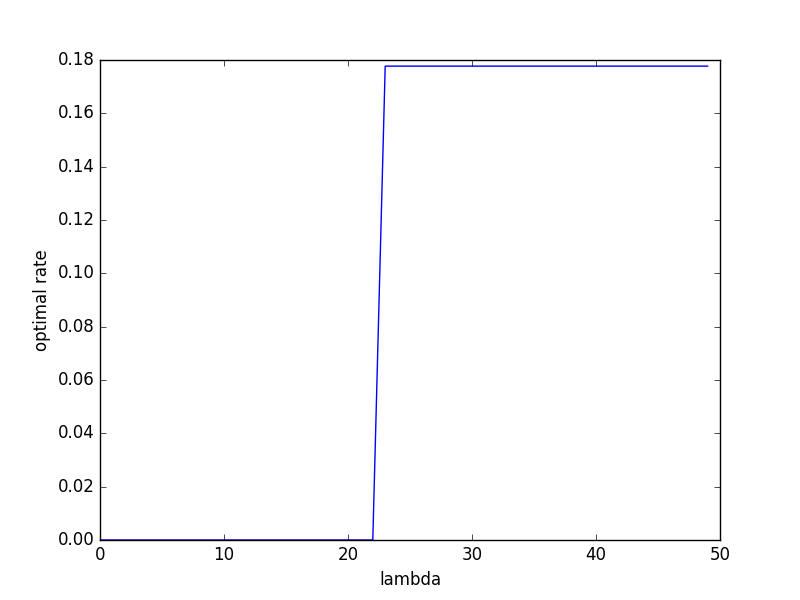}}
        \subfloat[$w(r, \lambda)$]{\includegraphics[width=4cm]{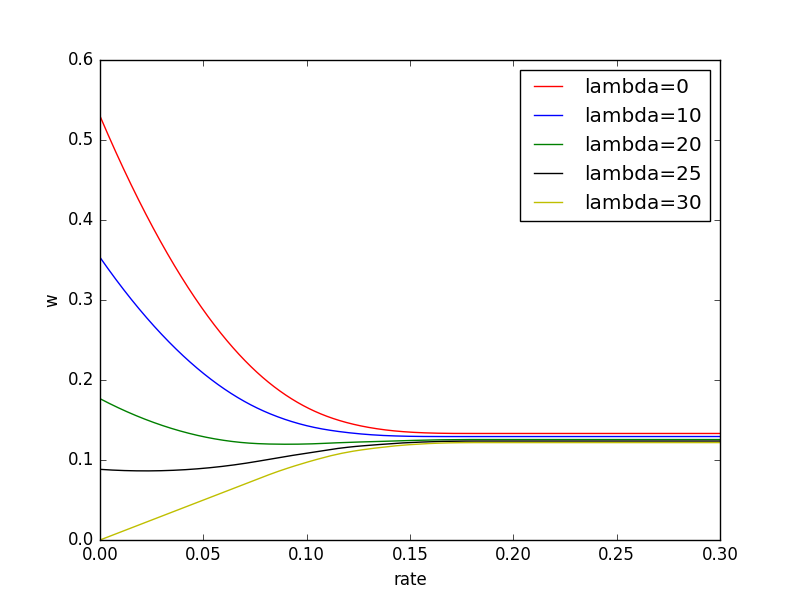}}
        \subfloat[$\mu_i-\sigma_i^2$]{\includegraphics[width=4cm]{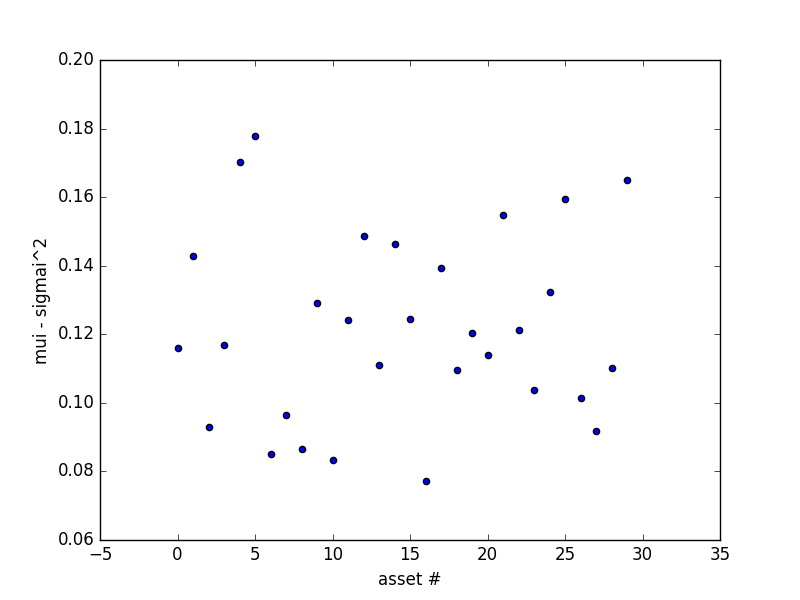}}
        \caption[ ]
        {\small \textit{Optimal interest rate with $N=30$ uncorrelated assets.  with mean and standard deviation $\mu_i,\sigma_i, i = 1,\ldots,N$ i.i.d uniform on $[0.1,0.2]$} } 
        \label{optR_diffMu_sigma}
    \end{figure*}

    \begin{figure*}[htbp]
        \centering
        \subfloat[optimal interest rate]{\includegraphics[width=4cm]{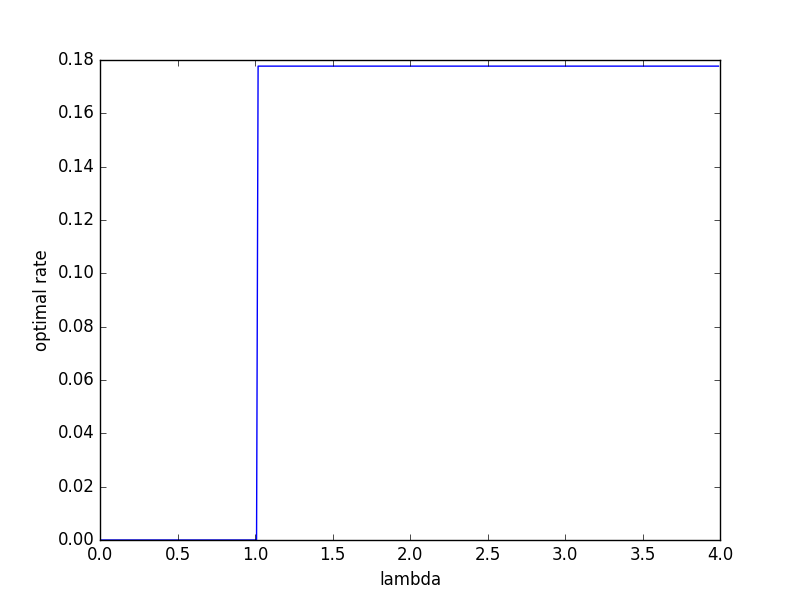}}
        \subfloat[$w(r, \lambda)$]{\includegraphics[width=4cm]{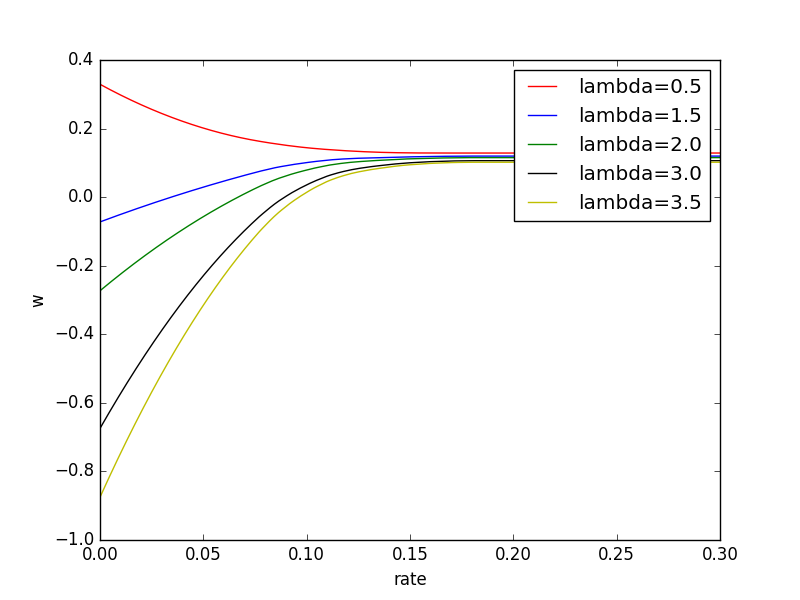}}
        \subfloat[$\mu_i - \sigma_i^2$]{\includegraphics[width=4cm]{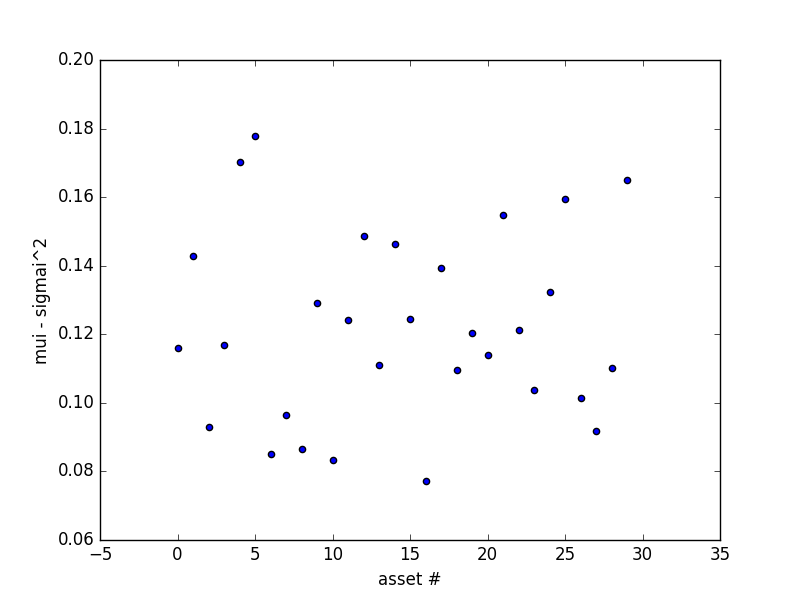}}

        \caption[ ]
        {\small \textit{Optimal interest rate with $N=30$ portfolio process, with correlation $\rho_0=0.8$ and $\mu_i,\sigma_i, i = 1,\ldots,N$ i.i.d uniform on $[0.1,0.2]$} } 
        \label{optR_diffMu_sigma_corrlAsset}
    \end{figure*}
    
    To illustrate the impact of risk aversion $\lambda$ on the optimal interest rate, we simulate three scenarios. First, in Figure~\ref{optR_fixedMu_sigma} we assume uncorrelated portfolio process, each with same mean and volatility $\mu_i=\sigma_i = 0.1$ for $i=1, \ldots , N$. This is Case (3.a), which is discussed above in this section.\\ 

Next, in Figure~\ref{optR_diffMu_sigma} we assume independent portfolio process but with mean and standard deviation $\mu_i,\sigma_i, i = 1,\ldots,N$ i.i.d uniform on $[0.1,0.2]$. This is Case (3.b), which is discussed above in this section. However, to our surprise, we observe the optimal interest rate to have only one jump as we increase the risk aversion parameter $\lambda$. \\

Finally, in Figure~\ref{optR_diffMu_sigma_corrlAsset} we assume correlated portfolio process with $\rho_0=0.8$ and mean and standard deviation $\mu_i,\sigma_i, i = 1,\ldots,N$ drawn from i.i.d uniform $[0.1,0.2]$. This is a generalized version of Case (3.c) discussed above. We observe that due to correlation in the portfolio process, even a relatively less risk averse central bank is forced to raise the interest rate.

\section{Long-term stability} In this section, we analyze the long-term behavior of the centered process: 
\begin{equation}
\label{eq:tilde-Y}
\tilde{Y} = \left(\tilde{Y}_1, \ldots, \tilde{Y}_N\right),\ \tilde{Y}_i(t) = Y_i(t) - \ol{Y}(t),\ i = 1, \ldots, N.
\end{equation}
It takes values in the hyperplane $\Pi := \{y \in \BR^N\mid y_1 + \ldots + y_N = 0\}$. In other words, we are trying to find whether log capitals of banks stay together as time $t$ goes to infinity, or they split into two or more ``clouds''. A similar problem was posed in \cite{BFK2005} and solved in \cite{Ichiba11} for rank-based models of financial markets. In that paper, log capitalizations of stocks are modeled as Brownian particles with drift and diffusion coefficients depending only on the current rank of this particle relative to other particles. This class of systems, known also as {\it first-order models} or {\it competing Brownian particles}, was a subject of much recent research. 

\smallskip

The key parameters are rates $c_{ij}$ of interbank cash flows. Under certain fairly general conditions on these rates, the process~\eqref{eq:tilde-Y} is ergodic: It has a unique statonary distribution;  and for any initial conditions, it converges to this distribution as $t \to \infty$. This section has two results. Theorem~\ref{thm:ergodic-graph} deals with the case of flow rates $c_{ij}$ being time-independent: $c_{ij}(t) \equiv c_{ij}$. Lemma~\ref{lemma:more-general} covers the general case. 

\smallskip

Assume the central bank has already chosen the interest rate $r = r^*$, as above. Then $\ol{Y}$ is a Brownian motion with drift coefficient $g(r^*)$ and diffusion coefficient $\rho^2(r^*)$. We have:
\begin{equation}
\label{eq:SDE-Y}
\md Y_i(t) = \md M_i(t) + \frac1{N}\SL_{j=1}^Nc_{ij}(t)\left(Y_j(t) - Y_i(t)\right)\,\md t,\ i = 1, \ldots, N.
\end{equation}
Here, the process: $M = (M_1, \ldots, M_N)$ is an $N$-dimensional Brownian motion with drift vector and covariance matrix
$\mu^* = (\mu^*_1, \ldots, \mu^*_N)$ and $A^* = (a_{ij}^*)_{i, j = 1,\ldots, N}$ from~\eqref{eq:optimal-drift} and~\eqref{eq:optimal-cov}. The centered process~\eqref{eq:tilde-Y} satisfies the SDE
\begin{equation}
\label{eq:SDE-tilde-Y}
\md\tilde{Y}_i(t) = \md\tilde{M}_i(t) + \frac1{N}\SL_{j=1}^Nc_{ij}(t)\left(\tilde{Y}_j(t) - \tilde{Y}_i(t)\right)\,\md t,\ i = 1, \ldots, N.
\end{equation}
Here, $\tilde{M}_i(t) := M_i(t) - \ol{M}(t)$ for $i = 1, \ldots, N$. Note that $\tilde{M} = (\tilde{M}_1, \ldots, \tilde{M}_N)$ is a $\Pi$-valued Brownian motion. It has drift vector 
\begin{equation}
\label{eq:new-drift}
\tilde{\mu}^* = (\tilde{\mu}^*_1, \ldots, \tilde{\mu}^*_N)',\ \tilde{\mu}^*_i := \mu^*_i - \ol{\mu}^*,\ \ol{\mu}^* := \frac1N\SL_{i=1}^N\mu_i^*
\end{equation}
and covariance matrix
\begin{equation}
\label{eq:new-cov}
\tilde{A}^* = (\tilde{a}^*_{ij}) := VA^*V,\ V = I_N - N^{-1}ee',\ e = (1, \ldots, 1)' \in \BR^N.
\end{equation}
Therefore, $\tilde{Y}$ is a Markov process. Denote by $P^t(x, \cdot)$ its transition function. Define the following measure norm  on $\Pi$ for a function $V : \Pi \to [1, \infty)$:
$$
\norm{\nu}_V := \sup\limits_{|f| \le V}\left|\int_{\Pi}f\md\nu\right|.
$$
We denote the Euclidean norm of a vector $x = (x_1, \ldots, x_d)' \in \mathbb R^d$ by
$$
\norm{x} := \left[x_1^2 + \ldots + x_d^2\right]^{1/2}.
$$

\begin{theorem}
Assume the flow rates $c_{ij}(t) = c_{ij}$ are constant. Define the graph $G$ on the set of vertices $\{1, \ldots, N\}$: $i \leftrightarrow j$ iff $c_{ij} > 0$. If $G$ is connected, then:

\smallskip

(a) $\tilde{Y}$ has a unique stationary distribution $\pi$ on $\Pi$, which is multivariate normal;

\smallskip

(b) the transition function satisfies for some constants $c, \la, k > 0$:
\begin{equation}
\label{eq:ergodic}
\norm{P^t(x, \cdot) - \pi(\cdot)}_{V} \le cV(x)e^{-kt},\ \ V(x) := \exp\left(\frac{\la}2\norm{x}^2\right);
\end{equation}

\smallskip

(c) for any bounded measurable function $f : \Pi \to \BR$ we have, a.s.:
$$
\lim\limits_{T \to \infty}\frac1T\int_0^Tf(\tilde{Y}(s))\,\md s = \int_{\Pi}f(y)\pi(\md y).
$$
\label{thm:ergodic-graph}
\end{theorem}

\begin{proof} From the properties of solutions of SDE and nondegeneracy of the covariance matrix $\tilde{A}^*$ of $M$, we have the following positivity property: 
\begin{equation}
\label{eq:positive}
P^t(x, C) > 0\ \mbox{for all}\ t > 0,\ x \in \Pi,\ C \subseteq \Pi\ \mbox{with}\ \mes_{\Pi}(C) > 0.
\end{equation}
The generator of $\tilde{Y}$ for all twice continuously differentiable functions $f : \Pi \to \BR$ is given by:
\begin{equation}
\label{eq:generator}
\CL f(x) := \left[\tilde{\mu}^* + \frac1{N}\CM x\right]\cdot \nabla f + \frac12\SL_{i=1}^N\SL_{j=1}^N\tilde{a}^*_{ij}\frac{\pa^2f}{\pa x_i\pa x_j}.
\end{equation}
Here, $\CM = (m_{ij})_{i, j = 1, \ldots, N}$ is the following matrix:
\begin{equation}
\label{eq:m-matrix}
m_{ij} = 
\begin{cases}
c_{ij},\ i \ne j;\\
-\SL_{k \ne i}c_{ik},\ i = j.
\end{cases}
\end{equation}
Now, plug this function $V$ from~\eqref{eq:ergodic} for a suitable $\la$ into the generator~\eqref{eq:generator}. Then \begin{equation}
\label{eq:derivatives-2}
\nabla V = \la xV,\ \frac{\pa^2 V}{\pa x_i\pa x_j} = \left(\la^2x_ix_j + \la\de_{ij}\right)V. 
\end{equation}
Combining~\eqref{eq:derivatives-2} with~\eqref{eq:generator}, we get:
\begin{equation}
\label{eq:estimate-preliminary}
\CL V = \left[\left(\la\tilde{\mu}^*\cdot x + \frac{\la}{N}x'\mathcal{M}x\right) + \frac12\left(\la^2(x'\tilde{A}^*x) + \la\tr(\tilde{A}^*)\right)\right]V.
\end{equation}
Using Lemma~\ref{lemma:estimate-M} below, we get:
\begin{equation}
\label{eq:negative-definite}
\frac1{N}\left[x'\CM x\right] \le -c_0\norm{x}^2,\ \ c_0 := \frac{c(\CM)}N.
\end{equation}
There exists a constant $a_0 > 0$ such that for all $x \in \Pi$, we have: $x'\tilde{A}^*x \le a_0\norm{x}^2$. Combining this observation with~\eqref{eq:estimate-preliminary} and ~\eqref{eq:negative-definite}, we get: 

\begin{equation}
\label{eq:estimate-new}
\CL V \le \left[\la\tilde{\mu}^*\cdot x - \la c_0\norm{x}^2 + \frac12a_0\la^2\norm{x}^2 + \frac12\la\tr(\tilde{A}^*)\right]\, V.
\end{equation}
Choose $\la := c_0/a_0$, then~\eqref{eq:estimate-new} takes the form
\begin{equation}
\label{eq:estimate-final}
\CL V(x) \le K(x)V(x),\ \ K(x) := \frac{c_0}{a_0}\tilde{\mu}^*\cdot x - \frac{c_0^2}{2a_0^2}\norm{x}^2 + \frac12\frac{c_0}{a_0}\tr(\tilde{A}^*).
\end{equation}
Note that, as $\norm{x} \to \infty$, we have: $K(x) \to -\infty$. Therefore, for some constants $c_1, c_2 > 0$, 
\begin{equation}
\label{eq:constants}
K(x) \le -c_1,\ \ \norm{x} \ge c_2.
\end{equation}
Recall the definition of the ball $\mathcal B(c_2)$ in~\eqref{eq:ball}. Since $\CL V$ and $V$ are continuous, we have:
\begin{equation}
\label{eq:other-constants}
\max\limits_{x \in \mathcal B(c_2)}\left[\CL V(x) + c_1V(x)\right] =: c_3  < \infty. 
\end{equation}
Combining~\eqref{eq:estimate-final} with~\eqref{eq:constants} and~\eqref{eq:other-constants}, we get:
\begin{equation}
\label{eq:final-new}
\CL V(x) \le -c_1V(x) + c_31_{\mathcal B(c_2)}(x). 
\end{equation}
Finally, combine~\eqref{eq:final-new} with the Feller property of $\tilde{Y}$ (i.e. for a bounded continuous function f, the map $x \mapsto P^tf(x)$ is also bounded and continuous for the transition function $P$ of $\tilde{Y}$), and with the positivity property~\eqref{eq:positive}. Apply Lemma~\ref{lemma:previous} from Appendix to Lebesgue reference measure $\psi$ and the function $V$ from~\eqref{eq:ergodic}. This completes the proof of (a) (the uniqueness of a stationary distribution), as well as of (b). The fact that this stationary distribution $\pi$ is multivariate normal follows from the observation that $\tilde{Y}$ is a multidimensional Ornstein-Uhlenbeck process on the hyperplane $\Pi$. 

\medskip

To finish the proof of Theorem~\ref{thm:ergodic-graph}, let us show (c): This is similar to the proof of \cite[Theorem 1]{Ichiba11}. Take any $r \ge c_2$. Adjusting the proof of~\eqref{eq:final-new} above, we find that there exists a positive constant $d(r)$ such that 
$$
\CL V(x) \le -c_1V(x) + d(r)1_{\mathcal B(r)}(x).
$$
Let $\tau_{\mathcal B(r)} := \inf\{t \ge 0 \mid \tilde{Y}(t) \in \mathcal B(r)\}$ be the hitting moment of the ball $\mathcal B(r)$, for a fixed $r > 0$. Apply \cite[Theorem 4.3(a)]{MT1993b}, with the function $V$ from~\eqref{eq:ergodic}, with $f := 1$,  $\de := 0$. Then 
$$
\ME_x\tau_{\mathcal B(r)} \le c_1^{-1}V(x),\ x \in \Pi.
$$
Use the fact that $V$ is bounded on compact subsets to verify assumption (b) in Lemma~\ref{lemma:Khas}. Assumption (a) of this lemma follows from the observation that the covariance matrix of $\tilde{Y}$ is constant. Now apply Lemma~\ref{lemma:Khas} from \cite[Theorem 4.1, Theorem 4.2]{KhasBook}, cited as \cite[Proposition 1]{Ichiba11}. This completes the proof of part (c) of Theorem~\ref{thm:ergodic-graph}. 
\end{proof}


\begin{lemma} Assume the flow rates are given by
$$
c_{ij}(t) = c_{ij}f(\tilde{Y}(t)),\ i, j = 1, \ldots, N,\ i \ne j;\ t \ge 0,
$$
where $f : \Pi \to \BR$ is a function such that 
$$
\varliminf\limits_{\norm{z} \to \infty}f(z) > 0,
$$
and $c_{ij}$ are real numbers as in Theorem~\ref{thm:ergodic-graph}. Then the conclusion of Theorem~\ref{thm:ergodic-graph} is the same, minus the conclusion that $\pi$ is multivariate normal. 
\label{lemma:more-general}
\end{lemma}

\begin{proof}
Similar to Theorem~\ref{thm:ergodic-graph}, but with the following changes: Instead of~\eqref{eq:estimate-preliminary}, we have:
$$
\CL V(x) = \left[\la\tilde{\mu}^*\cdot x + \frac{\la f(x)}{N}x'\mathcal Mx + \frac12\left(\la^2(x'\tilde{A}^*x) + \la\tr(\tilde{A}^*)\right)\right]V.
$$
There exist $c_4, c_5 > 0$ such that $f(x) \ge c_4$ for $x \in \Pi$, $\norm{x} \ge c_5$. Therefore, for such $x$, the estimate~\eqref{eq:estimate-final} is preserved with $c_0$ changed to $c_0c_4$. The rest of the proof is similar to that of Theorem~\ref{thm:ergodic-graph}. 
\end{proof}

Note, however, that if the graph $G$ is disconnected, then this stability breaks down. 
Indeed, assume $G$ has connected components $G_1$ and $G_2$ (only two for sake of notational simplicity; analysis is the same for more than two connected components), and the flow rates $c_{ij}$ are positive constants if $i$ and $j$ are adjacent, $c_{ij} = 0$ if not. By Theorem~\ref{thm:ergodic-graph}, we get: 
$$
(Y_i - \ol{Y}_1)_{i \in G_1},\ (Y_i - \ol{Y}_2)_{i \in G_2}
$$
are ergodic, that is, they satisfy an inequality similar to~\eqref{eq:ergodic}. Here, 
\begin{equation}
\label{eq:averages}
\ol{Y}_1(t) := \frac1{|G_1|}\SL_{i \in G_1}Y_i(t),\ \ \ol{Y}_2(t) := \frac1{|G_2|}\SL_{i \in G_2}Y_i(t).
\end{equation}
But these averages from~\eqref{eq:averages} are, in fact, Brownian motions with certain drift and diffusion coefficients, which are easy to calculate from~\eqref{eq:SDE-Y}. They are correlated, but not perfectly. Therefore, $\ol{Y}_1 - \ol{Y}_2$ is not ergodic, and the process $\tilde{Y}$ defined in~\eqref{eq:tilde-Y} is also not ergodic. Private banks are separated into two groups, which ``drift'' from each other.

\section{Concluding remarks} 

We studied a model of $N$ private banks exchanging money through interbank flows, borrowing from the general economy under an interest rate set by the central bank, and investing in portfolios consisting of risky assets; these portfolios are modeled by correlated geometric Brownian motions. This represents an enhancement of the model~\eqref{eq:basic-mean-field}, which is obtained in \cite{JP-Rene} as a result of banks borrowing from each other. We generalize the interbank flows from \cite{JP-Rene}, making them heterogeneous.  

\smallskip

Each private bank maximizes its expected terminal logarithmic utility. The central bank maximizes exponential utility function of the total size of the system. We are able to solve the control problems for each private banks and the central bank  because of this special choice of utility functions. The resulting dynamics looks a bit like~\eqref{eq:basic-mean-field}, except that each private bank has its own growth rate and volatility in the driving Brownian motion, and the flow rates $c_{ij}$ depend on $i$ and $j$. 

\smallskip

Our setup allowed us to study systemic risk and distribution of defaults under different market and investment scenarios. We also observe common economic phenomena of liquidity trap (where the monetary policy fails to boost the investment in risky assets) naturally arising from the model.

\smallskip

For future research, one can consider the case when some but not all portfolios $S_i$ satisfy $\mu_i \le \si_i^2$ (and are therefore unprofitable). In addition, it might be interesting to consider different utility functions for private banks, for example power utility. Since the corresponding Hamilton-Jacobi-Bellman equations are likely to be intractable, the problem might be analyzed using mean-field formulation, each bank is competing against the ``mass of banks''.

\section{Appendix} Let us state explicitly convergence results for general continuous-time Markov processes, used in the proof of Theorem~\ref{thm:ergodic-graph}. These results from classic papers \cite{DMT1995, MT1993a, MT1993b} link Lyapunov functions with  long-term convergence. In \cite[Lemma 2.3, Theorem 2.6]{MyOwn10}, we reformulate these results to make them more convenient for our use. Let us restate these results here, for convenience of the reader.

\begin{lemma} Take a Feller continuous strong Markov process $\mathfrak X = (\mathfrak X(t),\, t \ge 0)$ on the metric state space $\mathcal E$, with transition function $P^t(x, \cdot)$, and generator $\CL$. Denote by $\MP_x$ the probability measure under which $\mathfrak X(0) = x$. Assume for some positive reference measure $\psi$ and a function $V : \mathcal E \to [1, \infty)$ in the domain $\mathcal D(\CL)$ of the generator $\CL$, we have:

\smallskip

(a) for some compact subset $C \subseteq \mathcal E$, we have $\psi(C) > 0$;

\smallskip

(b) for all $\psi$-positive subsets $A \subseteq \mathcal E$, $x \in \mathcal E$, $t > 0$, we have: $P^t(x, A) > 0$. 

\smallskip

(c) for some constants $b, k > 0$ and a compact set $K \subseteq \mathcal E$, we have:
$$
\CL V(x) \le -kV(x) + b1_K(x),\ x \in \mathcal E; \ \mbox{and}\ \sup\limits_{x \in K}V(x) < \infty.
$$
Then there exists a unique stationary distribution $\pi$, and the transition function satisfies the following estimate: for some constants $D, \vk > 0$, 
$$
\norm{P^t(x, \cdot) - \pi(\cdot)}_V \le DV(x)e^{-\vk t},\ x \in \mathcal E,\ t \ge 0.
$$
\label{lemma:previous}
\end{lemma}

The following Strong Law of Large Numbers is taken from \cite[Theorem 4.1, Theorem 4.2]{KhasBook}. It holds under an assumption  which can be called {\it uniform positive recurrence}, and which can be deduced from existence of Lyapunov functions. Assume that $\mathcal E = \BR^d$ above, and $\mathcal X$ is the solution of an SDE with a certain drift vector, and the covariance matrix $\mathcal A(\cdot)$. Let $\tau_C := \inf\{t \ge 0\mid \mathfrak X(t) \in C\}$ be the hitting time of a subset $C \subseteq \BR^d$. Assume there exists a unique stationary distribution $\pi$. 

\begin{lemma} Assume for some open bounded domain $D \subseteq \CE$ with $C^2$ boundary, we have:

\smallskip

(a) the smallest eigenvalue of $\mathcal A(x)$ for $x \in \ol{D}$ is uniformly bounded away from zero;

\smallskip

(b) for every compact subset $\CK \subseteq \BR^d$, we have: $\sup\limits_{x \in \mathcal K}\ME_x\tau_D < \infty$. 

\smallskip

Then $\MP_x$-a.s. for every $x \in \BR^d$ and  bounded measurable function $f : \BR^d \to \BR$, we have:
$$
\lim\limits_{T \to \infty}\frac1T\int_0^Tf(\mathfrak X(t))\,\md t = \int_{\BR^d}f(x)\pi(\md x).
$$
\label{lemma:Khas}
\end{lemma}

\begin{lemma} Fix $\mu \in \mathbb R$ and $\sigma > 0$. Take a function $h : \mathbb R \to \mathbb R$, defined as 
$$
h(x) = \mu x - \frac{\sigma^2}2 x^2 - r (x-1)_{+}.
$$
Its global maximum is reached at the point $x^*$ and is equal to $h^* = h(x^*)$, where:

\begin{alignat*}{2}
    & \begin{aligned} & h(x^*) := 
\begin{cases}
r + \frac{(\mu - r)^2}{2\si^2},\ \ \mu -\si^2\ge r ;\\
\mu - \frac{\sigma^2}{2} ,\ \ \ 0 \le \mu -\si^2 \le r ; \\
\frac{\mu}{2\si^2} ,\ \ \mu \le \si^2.
\end{cases}
  \end{aligned} 
    & \hskip 3em &
  \begin{aligned}
  & x^* := 
\begin{cases}
\frac{\mu - r}{\si^2},\ \ \mu -\si^2\ge r ;\\
1 ,\ \ 0 \le \mu -\si^2 \le r ; \\
\frac{\mu}{\si^2} ,\ \ \mu \le \si^2.
\end{cases}
  \end{aligned}
\end{alignat*}

\label{lemma:solutionAlpha}
\end{lemma}

\begin{proof} We can write
$$
h(x) = 
\begin{cases}
\mu x - \frac{\sigma^2}2x^2 - r(x-1),\, x \ge 1;\\
\mu x - \frac{\sigma^2}2x^2,\, x \le 1.
\end{cases}
$$
First, note that the function $h$ is smooth everywhere except $x = 1$, and $h''(x) < 0$ for all $x \ne 1$. Therefore, if $h'(x) = 0$, then $h$ has a local maximum at $x$. Take derivatives on both intervals $(-\infty, 1]$ and $[1, \infty)$:
\begin{align*}
x \ge 1\quad \mbox{implies}\quad h'(x) &= (\mu-r) - \sigma^2x = 0  \implies x = x_1 :=\frac{\mu-r}{\sigma^2};\\
x \le 1\quad \mbox{implies}\quad h'(x) &= \mu - \sigma^2x = 0  \implies x = x_2 :=\frac{\mu}{\sigma^2}.
\end{align*}
On both these rays, $h$ is a parabola with branches facing down. 

Case 1. $ \mu -\si^2 \ge r$. Then $x_1,\, x_2 \ge 1$. Therefore, $h$ reaches maximum on $[1, \infty)$ at $x_1$, and on $(-\infty, 1]$ at $1$. Since $h$ reaches its maximum on $[1, \infty)$ at $x_1$ and not $1$, we have: $h(1) \le h(x_1)$. As a result, $x^* = x_1$. 

\smallskip

Case 2. $0 \le \mu -\si^2 \le r$. Then $x_1 \le 1$, but $x_2 \ge 1$. Therefore, $h$ reaches maximum on $[1, \infty)$ at $1$, and on $(-\infty, 1]$ at $1$. As a result, the global maximum will be at $x^*=1$.

\smallskip

Case 3. $ \mu -\si^2 \le 0$. Then $x_1, x_2 \le 1$. Therefore, $h$ reaches maximum on $[1, \infty)$ at $1$, and on $(-\infty, 1]$ at $x = x_2$. Similarly to Case 1, the global maximum is reached at $x^* = x_2$. 
\end{proof}

\begin{lemma} For the matrix $\CM$ defined in~\eqref{eq:m-matrix}, there exists a constant $c(\CM) > 0$ such that 
\begin{equation}
\label{eq:estimate-M}
x'\CM x \le -c(\CM)\norm{x}^2,\ \ x \in \Pi.
\end{equation}
\label{lemma:estimate-M}
\end{lemma}

\begin{proof} Note that $\CM$ is a generator matrix for a continuous-time Markov chain $Q = (Q(t),\, t  \ge 0)$ on $\{1, \ldots, N\}$. This Markov chain can be viewed as a biased random walk on the graph $G$: As it wants to jump out of a state $i$, it chooses one of its nearest neighbors $j$, such that $i$ and $j$ are connected, only not with uniform probability. This graph $G$ is connected. Therefore, this Markov chain is irreducible. Since it is finite, it is positive recurrent. From the standard results on continuous-time Markov chains, see for example \cite[Theorem 2.7.15]{Cambridge}, this Markov chain $Q$ has a unique stationary distribution 
$$
\pi^Q = \begin{bmatrix}\pi^Q_1& \ldots & \pi^Q_N\end{bmatrix}
$$
This stationary distribution satisfies $\pi^Q\CM = 0$. But the columns of the matrix $\CM$ sum up to zero. Therefore, $e'\CM = 0$, and $e/N$ is a stationary distribution. By uniqueness, $\pi^Q = e/N$. Let $\la_1, \ldots \la_N$ and $v_1, \ldots, v_N$ be the eigenvalues and eigenvectors of the matrix $\CM$
\begin{equation}
\label{eq:eigenvalues}
\CM v_i = \la_i v_i,\, i = 1, \ldots, N.
\end{equation}
The eigenvectors of the matrix $\CM$ are all real, because $\CM$ is symmetric. Next, nonzero eigenvalues are negative: This follows from \cite[Exercise 8.1]{IosifescuBook}. Each zero eigenvalue $\la_i$ has eigenvector $v_i$ which satisfies $v_i'\CM = 0$, that is, $v_i'$ is proportional to a stationary distribution. But the stationary distribution is unique, so we have (without loss of generality): 
$$
\la_1 = 0;\ \la_2, \ldots, \la_N < 0;\ v_1 = ce\ \mbox{for some constant}\ c. 
$$
Now, take an $x \in \BR^N$. Assume $v_1, \ldots, v_N$ are normalized: $\norm{v_i} = 1,\, i = 1, \ldots, N$. Because $\CM$ is symmetric, $v_1, \ldots, v_N$ form an orthonormal basis in $\BR^N$. Therefore, we can decompose
\begin{equation}
\label{eq:decomposition}
x = (x\cdot v_1)v_1 + (x\cdot v_2)v_2 + \ldots + (x\cdot v_N)v_N.
\end{equation}
For a vector $x \in \Pi$, we have: $x\cdot e = 0$, and therefore $x\cdot v_1 = 0$. Thus,~\eqref{eq:decomposition} takes the form 
\begin{equation}
\label{eq:truncated}
x = (x\cdot v_2)v_2 + \ldots + (x\cdot v_N)v_N.
\end{equation}
Apply the matrix $M$ to this vector in~\eqref{eq:truncated} and use~\eqref{eq:eigenvalues}. We have:
\begin{equation}
\label{eq:applied-M}
\CM x = (x\cdot v_2)\la_2v_2 + \ldots + (x\cdot v_N)\la_Nv_N.
\end{equation}
From~\eqref{eq:truncated} and~\eqref{eq:applied-M}, since $v_1, \ldots, v_N$ form an orthonormal basis of $\BR^N$, we have: 
\begin{equation}
\label{eq:M-product}
x'\CM x = \CM x\cdot x = \la_2(x\cdot v_2)^2 + \ldots + \la_N(x\cdot v_N)^2.
\end{equation}
In addition, multiplying~\eqref{eq:truncated} by itself, we get:
\begin{equation}
\label{eq:truncated-square}
\norm{x}^2 = x\cdot x = (x\cdot v_2)^2 + \ldots + (x\cdot v_N)^2.
\end{equation}
Let $c(\CM) := \min\left(|\la_2|, \ldots, |\la_N|\right) > 0$. Comparing~\eqref{eq:M-product} and~\eqref{eq:truncated-square}, we get~\eqref{eq:estimate-M}.
\end{proof}

\end{document}